\pdfoutput=1
\newif\ifFull
\Fulltrue
\documentclass[11pt]{article}
\advance \topmargin by -\headheight
\advance \topmargin by -\headsep
\evensidemargin \oddsidemargin
\usepackage{graphicx}
\usepackage{cite}
\usepackage[noend]{algorithmic}
\usepackage{times}
\usepackage{compress}

\usepackage{latexsym}

\usepackage{amsmath}

\usepackage{url}
\makeatletter
\def\@begintheorem#1#2{\sl \trivlist \item[\hskip \labelsep{\bf #1\ #2:}]}
\def\@opargbegintheorem#1#2#3{\sl \trivlist
      \item[\hskip \labelsep{\bf #1\ #2\ #3:}]}
\makeatother

\DeclareMathOperator{\poly}{poly}
\DeclareMathOperator{\E}{\bf E}
\newcommand{\eat}[1]{{}}

\newenvironment{proof}{\noindent{\bf Proof:}}{\hspace*{\fill}\rule{6pt}{6pt}\bigskip}

\newtheorem{theorem}{Theorem}
\newtheorem{lemma}[theorem]{Lemma}

\newtheorem{definition}{Definition}

\title{\huge Fully De-Amortized Cuckoo Hashing for\\
       \huge Cache-Oblivious Dictionaries and Multimaps}
\author{
Michael T. Goodrich \\[3pt]
Dept. of Computer Science \\ 
University of California, Irvine \\
Irvine, CA 92697-3435 USA \\
\textsf{goodrich(at)acm.org} \\[5pt]
\and 
Daniel S. Hirschberg \\[3pt]
Dept. of Computer Science \\ 
University of California, Irvine \\
Irvine, CA 92697-3435 USA \\
\textsf{dan(at)ics.uci.edu} \\[5pt]
\and
Michael Mitzenmacher\\[3pt]
Dept. of Computer Science \\ 
Harvard University \\
Cambridge, MA 02138 \\
\textsf{michaelm(at)eecs.harvard.edu}
\and
Justin Thaler \\[3pt]
Dept. of Computer Science \\ 
Harvard University \\
Cambridge, MA 02138 \\
\textsf{jthaler(at)fas.harvard.edu}
}

\date{\relax}

\begin{document}

\maketitle

\begin{abstract}
A \emph{dictionary} (or \emph{map}) 
is a key-value store that requires all keys be
unique, and
a \emph{multimap} is a key-value store that allows for multiple
values to be associated with the same key.
We design hashing-based indexing schemes for dictionaries and multimaps
that achieve worst-case optimal performance for lookups and updates, with
a small or negligible probability the data structure will require
a rehash operation, depending on whether we are working in the
the external-memory (I/O) model or one of the well-known versions
of the Random Access Machine (RAM) model.
One of the main features of our constructions is that they
are \emph{fully de-amortized}, meaning that their performance bounds
hold without one having to tune their constructions with certain performance
parameters, such as
the constant factors in the exponents of failure probabilities
or, in the case of the external-memory model,
the size of blocks or cache lines and the size of internal memory
(i.e., our external-memory algorithms are \emph{cache oblivious}).
Our solutions are based on a fully de-amortized
implementation of cuckoo hashing, which may be of independent interest.
This hashing scheme uses two cuckoo hash tables, one ``nested''
inside the other, with one serving as a primary structure 
and the other serving as an auxiliary supporting queue/stash structure that is 
super-sized with respect to traditional 
auxiliary structures but nevertheless adds
negligible storage to our scheme.
This auxiliary structure allows the success probability for cuckoo
hashing to be very high, which is useful in cryptographic or data-intensive
applications.
\ifFull


\fi
\end{abstract}

\thispagestyle{empty}
\setcounter{page}{0}
\clearpage

\vspace{-3mm}
\section{Introduction}
A \emph{dictionary} (or \emph{map})
is a key-value store that requires all keys be unique and
a \emph{multimap}~\cite{agmt-emm-11}
is a key-value store that allows for multiple
values to be associated with the same key.
Such structures are ubiquitous in the ``inner-loop''
computations involved in various algorithmic applications.
Thus, we are interested in implementations of these
abstract data types that are based on hashing and use $O(n)$ words of storage,
where $n$ is the number of items in the dictionary or multimap.

In addition, because such solutions
are used in real-time applications,
we are interested in implementations that are \emph{de-amortized},
meaning that they have asymptotically optimal worst-case 
lookup and update complexities, 
but may have small probabilities of overflowing their memory spaces.
Moreover, we would like these lookup and update bounds to 
hold without requiring that we build such a data structure 
specifically ``tuned'' to certain performance parameters, since it is
not always possible to anticipate such parameters at the time such a
data structure is deployed (especially if the length $p(n)$ of the sequence of 
operations on the structure is not known in advance).
For instance, if we wish for a failure probability that is bounded by
$1/n^c$ or $1/n^{c\log n}$, for some constant $c>0$,
we should not be required to build the data structure
using an amount of space or other components that are parameterized by $c$.
(For example, our first construction gives a single algorithm parameterized only by $n$, and for which, for \emph{any} $s>0$, lookups take time at most $s$ with
probability that depends on $s$. Previous constructions are parameterized by $c$ as well, and lookups take
time $c$ with probability $1-1/n^{c'}$ for fixed constants $c, c'$.)
Likewise, in the external-memory model, we would like
solutions that are \emph{cache-oblivious}~\cite{flpr-coa-99},
meaning that they achieve their performance bounds
without being tuned for the parameters of the memory hierarchy, like
the size, $B$, of disk blocks, or the size, $M$, of internal memory.
We refer to solutions that avoid such parameterized constructions as
\emph{fully de-amortized}.

By extending and combining various ideas in the algorithms literature, we show 
how to develop fully de-amortized
data structures, based on hashing, for dictionaries and multimaps.
Without specifically tuning our structures to constant factors in the
exponents, we provide solutions with performance bounds that hold
\emph{with high probability}, meaning that they hold with
probability $1-{1}/{\poly(n)}$, or \emph{with overwhelming probability},
meaning that they hold with probability $1-1/n^{\omega(1)}$. 
We also use the term 
\emph{with small probability} to mean ${1}/{\poly(n)}$,
and \emph{with negligible probability} to mean $1/n^{\omega(1)}$.
\ifFull
Briefly, we are able to achieve the following bounds:
\begin{itemize}
\item For dictionaries, we present two fully de-amortized constructions. 
The first 
is for the Practical RAM model~\cite{m-lbstr-96},
and it performs all standard operations 
(lookup, insert, delete) in $O(1)$ steps in the worst case, 
where all guarantees hold for any polynomial
number of operations  
{\em with high probability}. The second works in 
the external-memory (I/O) model, the standard RAM model,
and the AC$^0$ RAM model~\cite{thorup07},
and also achieves O(1) worst-case operations, where these guarantees hold
{\em with overwhelming probability}.
\item 
For multimaps, we provide a fully de-amortized scheme that in addition to the standard
operations can quickly return or delete {\em all} 
values associated with a key.  
Our construction
is suitable for external memory and is cache oblivious in this
setting. For instance, our external-memory solution returns all $n_k$ values
associated with a key in $O(1+n_k/B)$ I/Os, where $B$ is the block size,
but it is not specifically tuned with respect to the parameter $B$.
\end{itemize}
\fi
Our algorithms use linear space and work in the online setting, where
each operation must be completed before performing the next.

\ifFull
Our solutions for dictionaries and multimaps include the design of a
variation on \emph{cuckoo hash tables}, which were presented by Pagh
and Rodler~\cite{pr-ch-04} and studied by a variety of other
researchers (e.g.,
see~\cite{devroye2003cuckoo,k-ivcha-10,naor-history}).  These
structures use a freedom to place each key-value pair in one of two
hash tables to achieve worst-case constant-time lookups and removals
and amortized constant-time insertions, where operations fail with
polynomially small probability.  
We obtain the first fully de-amortized variation on cuckoo hash tables,
with negligible failure probability to boot.
\fi

\vspace{-3mm}
\subsection{Motivations and Models}
Key-value associations are used in many applications,
and hash-based dictionary schemes are well-studied in the 
literature (e.g., see~\cite{clrs-ia-01}).
Multimaps~\cite{agmt-emm-11} 
are less studied, although a multimap can be
viewed as a dynamic \emph{inverted file} or \emph{inverted index}
(e.g., see Knuth~\cite{k-ss-73}).
\ifFull
Given a collection, $\Gamma$, of
documents, an inverted file is an indexing strategy that allows one to
list, for any word $w$, all the documents in $\Gamma$ where $w$
appears.  
Multimaps also provide a natural representation framework for adjacency 
lists of graphs, with nodes being keys and 
adjacent edges being values associated with a key.  
For other applications, please see Angelino {\it et al.}~\cite{agmt-emm-11}.  
\fi


Hashing-based implementations of dictionaries and multimaps must
necessarily be designed in a computational model that supports
indexed addressing, such as the Random Access Machine (RAM) model
(e.g., see~\cite{clrs-ia-01}) or the external-memory (I/O) model
(e.g., see~\cite{av-iocsrp-88,v-emads-01}).
Thus, our focus in this paper is on solutions in such models.
There are, in fact, several versions of the RAM model\ifFull, and,
rather than insist that our solutions be implemented in a specific
version, we consider solutions for several of the most well-known
versions\fi:
\begin{itemize}
\item
\emph{The standard RAM}:
all arithmetic and comparison operations, 
including integer addition, subtraction, multiplication, and division,
are assumed to run in $O(1)$ time\footnote{Most, but not all, 
	researchers also assume that the standard RAM supports 
	bit-wise Boolean operations, such as AND, OR, and XOR; hence, we 
	also allow for these as constant-time operations 
	in the standard RAM model.}.
\item
\emph{The Practical RAM}~\cite{m-lbstr-96}:
integer addition, as well
as bit-wise Boolean and SHIFT operations on words, are assumed to run in
$O(1)$ time, but not integer multiplication and division.
\item
\emph{The AC$^0$ RAM}~\cite{thorup07}:
any AC$^0$ function can be performed in
constant time on memory words, 
including addition and subtraction, as
well as several bit-level operations included in the instruction sets
of modern CPUs, such as Boolean and SHIFT operations.
This model does not allow for constant-time multiplication, however, since
multiplication is not in AC$^0$~\cite{fredmanwillard}.
\end{itemize}
Thus, the AC$^0$ RAM is arguably the most realistic, the standard RAM
is the most traditional,
and the Practical RAM is a restriction of both.
As we are considering dictionary and multimap solutions in all of these
models, we assume that the hash functions being used are 
implementable in the model in question and that they run in
$O(1)$ time and are sufficiently random to support cuckoo hashing.
This assumption is supported in practice, for instance, by the fact that 
one of the most widely-used 
hash functions, SHA-1, can be implemented in $O(1)$ time in
the Practical RAM model. See also Section \ref{sec:extensions}
for discussion of the theoretical foundations of this assumption.

A framework that is growing in interest and impact for designing 
algorithms and data structures
in the external-memory model is the \emph{cache-oblivious}
design paradigm, introduced by Frigo
{\it et al.}~\cite{flpr-coa-99}.  
In this external-memory paradigm, one designs
an algorithm or data structure to minimize the number of I/Os
between internal memory and external memory, but the algorithm must not be
explicitly parameterized by the block size, $B$, or the
size of internal memory, $M$.  The advantage of this
approach is that one such algorithm can comfortably scale across all levels
of the memory hierarchy and can also be a better
match for modern compilers that perform predictive memory fetches.

Our notion of a ``fully de-amortized'' data structure extends the
cache-oblivious design paradigm in two ways.
First, it requires that all operations be characterized in terms of
their worst-case performance, not its amortized performance. Second, it extends to internal-memory models
the notion of avoiding specific tuning of
the data structure in terms of non-essential parameters.
Formally, we say that
a data structure is \emph{fully de-amortized} if its performance
bounds hold in the worst case and the only parameter 
its construction details depend is $n$, the number of items it stores.
Thus, a fully de-amortized data structure implemented 
in external-memory
is automatically cache-oblivious and its I/O bounds hold
in the worst case.

In addition to being fully de-amortized,
our strongest constructions have performance bounds that hold with
overwhelming probability. 
While our aim of achieving structures that provide
worst-case constant time operations with overwhelming
probability instead of with high probability may seem like a
subtle improvement, there are many applications where it is essential.
In particular, it is common in cryptographic applications to aim for
negligible failure probabilities.
For example, cuckoo hashing structures 
with negligible failure probabilities
have recently found applications in oblivious RAM
simulations~\cite{GoodrichMitzenmacherICALParXivVersion}. 
Moreover, a significant motivation for de-amortized 
cuckoo hashing is
to prevent timing attacks and clocked adversaries 
from compromising a system \cite{arbitman2009amortized}. 
Finally, guarantees that hold with overwhelming 
probability allow us to handle super-polynomially
long sequences of updates, as long as the total 
number of items resident in the dictionary
is bounded by $n$ at all times. 
This may be useful in long-running or data-intensive applications.

\vspace{-3mm}
\subsection{Previous Related Work}
Since the introduction of the
cache-oblivious framework by
Frigo {\it et al.}~\cite{flpr-coa-99},
several cache-oblivious algorithms have subsequently been presented,
including cache-oblivious B-trees~\cite{bdf-cob-00},
cache-oblivious binary search trees~\cite{bfj-cost-02},
and cache-oblivious sorting~\cite{bfv-ecos-08}.
Pagh {\it et al.}~\cite{pwyz-coh-10}
describe a scheme for cache-oblivious hashing, which is based on
linear probing and achieves $O(1)$ expected-time performance for
lookups and updates, but it does not achieve constant
time bounds for any of these operations in the worst case.

As mentioned above,
the multimap ADT is related to the inverted file and inverted
index structures, which are well-known in 
text indexing applications (e.g., see Knuth~\cite{k-ss-73})
and are also used in
search engines (e.g., see Zobel and Moffat~\cite{zm-iftse-06}).
Cutting and Pedersen~\cite{cp-odiim-90} describe an inverted file
implementation that uses B-trees for the indexing structure and
supports insertions, but doesn't support
deletions efficiently.
More recently, Luk and Lam~\cite{ll-eimef-07} describe an internal-memory
inverted file implementation based on hash tables with chaining, but
their method also does not support fast item removals.
Lester {\it et al.}~\cite{lmz-eoict-08,lzw-eoimc-06} 
and B\"{u}ttcher {\it et al.}~\cite{bcl-himgt-06}
describe external-memory
inverted file implementations that support item insertions only.
B\"{u}ttcher and Clarke~\cite{bc-itvqt-05}
consider trade-offs for allowing for
both item insertions and removals,
and Guo {\it et al.}~\cite{gcxw-eoli-07}
give a solution for performing such operations
by using a B-tree variant.
Finally, Angelino {\it et al.}~\cite{agmt-emm-11} describe an
efficient external-memory data structure for the multimap ADT, but like the
above-mentioned work on inverted files, their method is not
cache-oblivious; hence, it is not fully de-amortized.

Also as mentioned above, our solutions 
include the design of a variation on \emph{cuckoo hash tables},
which were presented by Pagh and Rodler~\cite{pr-ch-04}
and studied by a variety of other researchers 
(e.g., see~\cite{devroye2003cuckoo,k-ivcha-10,naor-history}).
These structures use a freedom to place 
each key-value pair in one of two hash tables to achieve worst-case
constant-time lookups and removals and amortized constant-time
insertions with high probability.
Kirsch and Mitzenmacher~\cite{km-uqdac-07} and
Arbitman {\it et al.}~\cite{arbitman2009amortized} study
a method for de-amortizing cuckoo hashing, which achieves
constant-time lookups, insertions, and deletions with high probability,
and uses space $(2+\epsilon) n$ for any constant $\epsilon > 0$ 
(as is standard in cuckoo hashing). 
These methods 
are not fully de-amortized, however, since, in order to achieve a failure
probability of $1/n^c$, 
they construct an auxiliary structure consisting of $O(c)$ small lookup tables.
In contrast, neither of our dictionary constructions are parameterized by $c$;
furthermore our second construction provides guarantees that hold with
overwhelming probability rather than with high probability.
In a subsequent paper,
Arbitman {\it et al.}~\cite{ans-bchcw-10} study
a hashing method that achieves worst-case constant-time lookups, insertions, 
and removals with high probability while maintaining 
loads very close to 1, 
but their method also is not fully de-amortized.

Kirsch, Mitzenmacher, and Wieder~\cite{kmw-chs-09}
introduced the notion of a stash for cuckoo hashing, which allows the
failure probability to be reduced to $O(1/n^{\alpha})$,
for any $\alpha > 0$,
by using a constant-sized adjunct memory to store items that
wouldn't otherwise be able to be placed.
One of our novel additions in this paper
is to demonstrate that by using super-constant sized stashes, along with
a variation of the q-heap data structure, we can ensure failures happen
only with negligible probability while maintaining constant time lookup
and delete operations.

\vspace{-3mm}
\subsection{Our Results}
In this paper we describe efficient hashing-based
implementations of the dictionary and multimap ADTs.
Our constructions are fully de-amortized and are alternately
designed for the external-memory (I/O) model and 
the well-known versions of the RAM model mentioned above.
Because they are fully de-amortized,
our external-memory algorithms are
cache-oblivious.

We begin by presenting two new 
fully de-amortized cuckoo hashing schemes in 
Section \ref{sec:cuckoo}. Both of our constructions provide 
$O(1)$ worst-case lookups, insertions, and deletions, where the guarantees
of the first construction (in the Practical RAM model) 
hold with high probability,
and the guarantees of the second construction (in the 
external-memory model, the standard RAM model,
and the AC$^0$ RAM model)
hold with overwhelming probability.
Moreover,
these results hold even if we use polylog($n$)-wise independent hash functions. 
Like the construction of Arbitman {\it et al.}~\cite{arbitman2009amortized}, 
both of our dictionaries use space $(2+\epsilon)n$ 
for any constant $\epsilon>0$,
though when combined with another result of 
Arbitman {\it et al.}~\cite{ans-bchcw-10}, 
we can achieve $(1+\epsilon)n$ words
of space for any constant $\epsilon > 0$ (see Section \ref{sec:extensions}).
Our second dictionary can be seen as 
a quantitative improvement over 
the previous solution
for de-amortized cuckoo tables~\cite{arbitman2009amortized}, 
as the guarantees of \cite{arbitman2009amortized} only
hold with high probability (and their solution is not fully
de-amortized).

Both of our dictionary constructions utilize a cuckoo hash table 
that has another cuckoo hash table
as an auxiliary structure.
This secondary cuckoo table functions simultaneously as an
operation queue (for the sake of 
de-amortization~\cite{arbitman2009amortized,ans-bchcw-10,km-uqdac-07})
and a stash (for the sake of improved probability of success \cite{kmw-chs-09}).

Our second construction also makes use of a data structure for the AC$^0$ RAM model we call the
\emph{atomic stash}. This structure maintains a small dictionary, of
size at most $O(\log^{1/2} n)$, so as to support constant-time
worst-case insertion, deletions, and lookups, using $O(\log^{1/2} n)$
space.  
This data structure is related to the q-heap or atomic heap
data structure of Fredman and Willard~\cite{fredmanwillard} (see
also~\cite{Willard}), which requires the use of lookup tables of size
$O(n^\epsilon)$ and is limited to sets of size $O(\log^{1/6} n)$.
Andersson {\it et al.}~\cite{Andersson1999337}
and Thorup~\cite{thorup07} give
alternative implementations of these data structures in the AC$^0$
RAM model,
but these methods still need
precomputed functions encoded in table lookups or have time bounds
that are amortized instead of worst-case.  
Our methods instead make no
use of precomputed lookup tables, hold in the worst case, and use
techniques that are simpler than those of Anderson 
{\it et al.}~and Thorup.  
We emphasize that our
results do not depend on our specific atomic stash implementation; one
could equally well use q-heaps, atomic heaps, or other data structures
that allow constant-sized lookups into data sets of size $\omega(1)$
(under suitable assumptions, such as keys fitting into a memory word)
in order to obtain bounds 
that hold with overwhelming probability in the standard or AC$^0$ RAM
models.  
We view this combination of cuckoo hash tables with atomic stashes (or
other similar data structures) as an important contribution, as they
allow super-constant sized stashes for cuckoo hashing while still
maintaining constant time lookups.

In Section \ref{sec:multimap}, we also show how to build on our 
fully de-amortized 
cuckoo hashing scheme
to give an efficient cache-oblivious multimap implementation
in the external memory model. 
Our multimap implementation utilizes two instances of the nested cuckoo
structure of Section \ref{sec:cuckoo},
together with arrays for storing collections of key-value pairs that
share a common key. 
This implementation assumes that there is a cache-oblivious mechanism
to allocate and deallocate power-of-two sized 
memory blocks with constant-factor space and I/O overhead; 
this assumption is theoretically justified by the results of 
Brodal {\it et al.}~\cite{buddysystem}.
A lower bound of Verbin and Zhang~\cite{verbin2010} 
implies that our bounds are optimal up to
constant factors in the external memory model, 
even if we did not support fast findAll and removeAll operations.

In addition to a theoretical analysis of our data structures, we have 
performed preliminary experiments with an implementation, and a later
writeup will include full details of these.

The time 
bounds we achieve for the dictionary and multimap ADT methods are shown in
Table~\ref{tbl:bounds}.
Our space bounds are all $O(n)$, for storing a dictionary or multimap
of size at most $n$.

\vspace{-3pt}

\begin{table*}[hbt]
\begin{center}
\small{
\begin{tabular}{|c|c|c|}
\hline
\textbf{Method} 
& \textbf{Dictionary I/O Performance} 
& \textbf{Multimap I/O Performance} \\
\hline
\rule[-2pt]{0pt}{15pt} add$(k,v)$ 
& $O(1)$ 
& $O(1)$ \\
\hline
\ifFull
\rule[-2pt]{0pt}{15pt} containsKey$(k)$ 
& $O(1)$ 
& $O(1)$ \\
\hline
\rule[-2pt]{0pt}{15pt} containsItem$(k,v)$ 
& $O(1)$ 
& $O(1)$ \\
\hline
\fi
\rule[-2pt]{0pt}{15pt} remove$(k,v)$ 
& $O(1)$ 
& $O(1)$ \\
\hline
\rule[-2pt]{0pt}{15pt} get$(k)$/getAll$(k)$ 
& $O(1)$ 
& $O(1+n_k/B)$ \\
\hline
\rule[-2pt]{0pt}{15pt} removeAll$(k)$ 
& --
& $O(1)$ \\
\hline
\end{tabular}
}
\end{center}
\small{
\vspace{-6pt}
\caption{\label{tbl:bounds} 
Performance bounds for our dictionary
and multimap
implementations, 
which all hold in the worst-case with overwhelming probability,
except for implementations in the Practical RAM model, in which case
the above bounds hold with high probability.
These bounds are asymptotically
optimal. 
We use $B$ to denote the block size, 
and
$n_k$ to denote the number of items with key equal to $k$. 
}
}
\end{table*}

\vspace{-8pt}
\section{Nested Cuckoo Hashing}
\label{sec:cuckoo}
In this section, we describe both of our nested cuckoo hash table data structures,
which provide fully de-amortized dictionary data structures with
worst-case $O(1)$-time lookups and removals and $O(1)$-time
insertions with high and overwhelming probability, respectively.
At a high level, this
structure is similar to that of Arbitman {\it et al.}~\cite{arbitman2009amortized} and
Kirsch and Mitzenmacher~\cite{km-uqdac-07}, in that our scheme and these
schemes use a cuckoo hash table as a primary structure, and an auxiliary queue/stash structure to de-amortize
insertions and reduce the failure probability. But
our approach substantially differs from prior methods
in the details of the auxiliary structure. In particular, our auxiliary structure is 
itself a full-fledged cuckoo hash table, with its own (much smaller) queue and stash,
 whereas prior methods use more
traditional queues for the auxiliary structure.

\ifFull
Our two dictionary constructions are identical, except for the implementation of the ``inner'' cuckoo hash table's queue and stash. 
Below, we describe both constructions simultaneously,
highlighting their differences when necessary.
\fi

\vspace{-2pt}
\subsection{The Components of Our Structure}
Our dictionaries maintain a dynamic set, $S$, of at most $n$
items.
Our primary storage structure is a cuckoo hash table, $T$, subdivided
into two subtables, $T_0$ and $T_1$, 
together with two random hash
functions, $h_0$ and $h_1$.
Each table $T_i$ stores at most one item in each of its $m$ cells,
and we assume $m \geq (1+\epsilon) n$ where $n$ is the total number of items,
for some fixed constant $\epsilon>0$.
For any item $x=(k,v)$ that is stored in $T$, $x$ will either be stored
in cell $T_0[h_0(k)]$
or in cell $T_1[h_1(k)]$.
Some of the items in $S$ may not be stored in $T$, however. They will
instead be stored in an auxiliary structure, $Q$.

The structure $Q$ is simultaneously two double-ended 
queues (deques), both of which support fast enqueue
and dequeue operations at either the front or the rear, 
and a cuckoo hash table with its own queue and stash. 
Because $Q$ is a cuckoo hash table,
it supports
worst-case constant-time lookups for items based on their keys.
The two deques correspond to the \emph{queue} and \emph{stash} of the primary
structure; we call them \text{OuterQueue} and \text{OuterStash} respectively. 
Each item $x$ that is stored in $Q$ is therefore also augmented with a prev
pointer, which points to the predecessor of $x$ in its deque, and a next
pointer, which points to the successor of $x$ in its deque.
We also augment both deques with front and rear pointers, which respectively
point to the first element in the deque and the last element in
the deque.

The ``inner'' cuckoo hash structure consists of two tables, $R_0$ and $R_1$, each having $m^{2/3}$
cells, together with two random hash functions, $f_0$ and $f_1$, as
well as a small list, $L$, which is used to implement two deques that we call \text{InnerQueue} and \text{InnerStash} respectively. 
Each item $x=(k,v)$ that is stored in $Q$ will be located 
in cell $R_0[f_0(k)]$
or in cell $R_1[f_1(k)]$, or it will be in the list $L$. The only manner in which our two constructions differ is 
in the implementation of $L$.

\ifFull
In summary, our dictionary data structure consists simply of the
``outer'' cuckoo table $T$ with a queue and stash, 
and the ``inner'' cuckoo table $Q$, which contains its own queue and stash.
\fi

\subsection{Operations on the Primary Structure}
To perform a lookup, that is, a get$(k)$,
we try each of the cells
$T_0[h_0(k)]$,
$T_1[h_1(k)]$,
$R_0[f_0(k)]$,
and $R_1[f_1(k)]$, and 
perform a lookup in $L$, until we
either locate an item, $x=(k,v)$,
or we determine that there is no item in these locations with key equal to
$k$, in which case we conclude that there is no such item in
$S$.

Likewise, to perform a remove$(k)$ operation, we first perform a
get$(k)$ operation. 
If an item $x=(k,v)$ is found in one of the locations
$T_0[h_0(k)]$ or
$T_1[h_1(k)]$,
then we simply remove this item from this cell.
If such an $x$ is found in
$R_0[f_0(k)]$,
or $R_1[f_1(k)]$, or in the structure $L$, then we remove $x$ from this
cell and we remove $x$ from its deque(s) by updating the prev and
next pointers for $x$'s neighbors so that they now point to each
other. 

Performing an add$(k,v)$ operation is somewhat more involved.
We begin by performing an enqueueLast$(x$, $0,$ $\text{OuterQueue})$ operation, which
adds $x=(k,v)$ at the end of the deque \text{OuterQueue},
together with the bit $0$ to indicate that $x$
should next be inserted in $T_0$.
Then, for constants $\alpha, \alpha' \ge 1$, set in the analysis, we perform
$\alpha$ (outer) \emph{insertion substeps}, followed by $\alpha'$ (inner) insertion substeps.
Each outer insertion substep begins by performing a dequeueFront(\text{OuterQueue})
operation to remove the pair $((k,v),b)$ at the front of the 
deque.
If the cell $T_b[h_b(k)]$ is empty, then we add $(k,v)$ to this cell
and this ends this substep.
Otherwise, we evict the current item, $y$, in 
the cell $T_b[h_b(k)]$, replacing it with $(k,v)$.
If this insertion substep just created a second cycle in the insertion
process (which we can detect by a simple marking scheme, using $O(\log^2 n)$ space with overwhelming probability)\footnote{\cite{arbitman2009amortized} refers
to this as a \emph{cycle-detection mechanism}, and notes there are many possible instantiations.
See also~\cite{kmw-chs-09}.}, 
then we complete the insertion substep by
performing an enqueueLast$(y,b', \text{OuterStash})$ operation, 
where $b'=(b+1)\bmod 2$.
If this insertion substep has not created a second cycle, however, then we
complete the insertion substep
by performing an enqueueFirst$(y,b', \text{OuterQueue})$ operation, 
where $b'=(b+1)\bmod 2$, which adds the pair $(y,b')$ to the front of the primary structure's queue.
Thus, by design, an add operation takes $O(\alpha)=O(1)$ time in the worst case, assuming the operations on $Q$ run in $O(1)$ time and succeed.

Finally, similar to \cite{yuriy_thesis}, every $m^{1/4}$ operations we try to insert an element from the stash
by performing a dequeueFront(\text{OuterStash})
operation to remove the pair $((k,v),b)$ at the front of the 
deque and then spending $2$ moves trying to insert $(k, v)$. If no free slot is found,
we return the current element from the connected component of $(k, v)$ to the front of the stash.
This ensures that items belonging to connected components that have had cycles
removed via deletions do not remain in the stash long after the deletions occur.
All constants mentioned throughout are chosen for theoretical convenience; no effort 
has been made to optimize them.

\subsection{Operations on the Auxiliary Structure}
As mentioned above,
the auxiliary structure, $Q$, is a standard cuckoo
table of size $m'=m^{2/3}$ augmented with its own (inner) queue and stash maintained via the structure $L$,
and pointers to give $Q$ the 
functionality of two double-ended queues, called \text{OuterQueue} and \text{OuterStash}.
The enqueue and dequeue operations to \text{OuterQueue} and \text{OuterStash}, therefore, involve standard
$O(1)$-time pointer updates to maintain the deque property,
plus insertion and deletion
algorithms for the inner cuckoo table. Our inner insertion algorithm is different from our outer one in that we do not immediately
place items on the back of the inner queue. Instead, on an insertion of item $(k, v)$, we spend 16 steps trying to insert $(k, v)$ into
the inner cuckoo table immediately. 
If the insertion does not complete in 16 steps, we place the current element from the connected component of $(k, v)$ in the back of \text{InnerQueue}.
The reason for this policy is that, as we will show, with overwhelming probability almost every item in the inner table can be
inserted in 16 steps, due to the extreme sparsity of the table.
Finally, every $m^{1/6}$ ``inner'' operations we additionally spend 16 moves trying to insert an element from the front of \text{InnerQueue} and
a single move trying to insert an element from the front of \text{InnerStash}, returning elements to the back of \text{InnerQueue}
or \text{InnerStash} in the event that a vacant slot has not been found.
The purpose of this is to ensure that items whose connected components have shrunk due to deletions do not remain
in the inner queue or stash unnecessarily. 


The only manner in which our two constructions differ is 
in the implementation of $L$, which supports the inner stash and inner queue.  In our first construction, for the Practical
RAM model, 
$L$ is simply two double-ended queues (one for \text{InnerQueue}, and one for \text{InnerStash}),
and when performing a lookup in $L$, we simply look at all the cells in both deques.
In our second construction, we implement $L$ using a data structure we call the \emph{atomic stash}.
This structure maintains a small dictionary, of size at most 
$O(\log^{1/2} n)$, so as to support constant-time worst-case insertion, deletions, and lookups, using 
$O(\log^{1/2} n)$ space. Thus, in our second construction, the structure $L$ is simultaneously
two dequeues and an atomic stash to enable fast lookups into $L$. Details are in Appendix \ref{app:stash}.

\section{Analysis}

The critical insights into why such a standard approach is sufficient
to support fully de-amortized constant-time updates and lookups 
for the auxiliary structure, $Q$, 
are based on enhancing the analysis of previous work
for cuckoo hash tables.

\begin{definition} A sequence $\pi$ of insert, delete, and lookup operations is \emph{n-bounded} if at any point in time 
during the execution of $\pi$ the data structure contains at most $n$ elements. \end{definition}

We prove the following two theorems. 
\begin{theorem} \label{thm:fastinsertsinner}
Let $C$ be the inner cuckoo hash table consisting
of two tables of size $m^{2/3}$, together with a queue/stash $L$ as described above. For any polynomial $p(n)$, any 
$m^{1/3}$-bounded sequence of 
operations $\pi'$ on $C$ of length $p(n)$, and any $s \leq m^{1/6}$, every insert into $C$ completes in O(1)
steps, and $L$ has size at most $s$, with probability at least $1-p(n)/m^{\Omega(\sqrt{s})}$.
\end{theorem}

Theorem \ref{thm:fastinsertsinner} states that with overwhelming probability, all insertions performed on the
auxiliary structure $Q$ will run in $O(1)$ time, and all of the inner table's internal data structures will stay small (ensuring fast lookups), provided we never
try to put more than $m^{1/3}$ items in $Q$.
The following theorem addresses this size condition.

\begin{theorem} \label{thm:Qsmall}
Let $T$ be a cuckoo hashing scheme consisting
of two tables of size $m$. Let $m \geq (1+\epsilon)n$, for some constant $\epsilon>0$,
and let $Q$ be the queue/stash as described above.
For any polynomial $p(n)$ and any $n$-bounded sequence of operations $\pi$ of length $p(n)$, the number of items stored in 
$Q$ will never be more than $2\log^6 n$,
with probability at least $1-1/n^{\Omega(\log n)}$.
\end{theorem}

\subsection{Notation and Preliminary Lemmata}
A standard tool in the analysis of cuckoo hashing is the \emph{cuckoo graph} $G$. Specifically, given a cuckoo hash table $T$,
subdivided into two tables $T_0$ and $T_1$, each with $m$ cells, and hash functions $h_0$ and $h_1$, the cuckoo graph
for $T$ is a bipartite multigraph (without self-loops), in which left vertices represent the table cells in $T_0$ and right vertices represent the cells
in $T_1$. Each key $x$ inserted into the hash table corresponds to an edge connecting 
$T_0[h_0(x)]$ to $T_1[h_1(x)]$. Thus, if $S$ denotes the set of $n$ items in $T$, and each of $T_0$ and $T_1$ have $m$ cells, the cuckoo graph $G(S, h_0, h_1)$ contains $2m$ nodes ($m$ on each side) and $n$ edges.
Given a cuckoo graph $G(S, h_0, h_1)$, we use $C_{S, h_0, h_1}(v)$ to denote the number of edges
in the connected component that contains $v$.

In our analysis, there will actually be \emph{two} cuckoo graphs $G$ and $G'$. $G$ corresponds to the outer table, and has vertex set $[m] \times [m]$ and at most $n$ edges
at any point in time.
$G'$ corresponds to the inner table and has vertex set $[m']\times[m']$ for $m'=m^{2/3}$, and (we will show) $n'\leq m^{1/3}$ edges with overwhelming probability at any point in time. 


We will make frequent use of the following lemmata.

\begin{lemma} \label{kmwlemma2}
Let $G(S, h_0, h_1)$ be a cuckoo graph with $n$ edges and vertex set $[m] \times [m]$. Then
\begin{equation*} \Pr(C_{S, h_0, h_1}(v) \geq k) \leq \left (\frac{nk}{m} \right )^k \frac{1}{k!}, \end{equation*}
where the probability is over the choice of $h_0$ and $h_1$.
\end{lemma}
\begin{proof}
By standard arguments, (e.g. \cite[Lemma 1]{DM03}), $\Pr(C_{S, h_0, h_1}(v) \geq k) \leq \Pr(\mbox{Bin}(nk,1/m) \geq k).$
The lemma follows by an easy calculation.
\end{proof}

We now present some basic facts about the stash.  The first relates the time it takes to determine
whether an element needs to be put in the stash on an insertion.

\begin{lemma} \label{moveslemma} \cite[Claim 5.4]{yuriy_thesis} For any element $x$ and any cuckoo graph $G(S, h_0, h_1)$, 
the number of moves required before $x$ is either placed in the stash or inserted into the cuckoo table is at most $2C_{S, h_0, h_1}(x)$.
\end{lemma}


\cite[Lemma 2.2]{kmw-chs-09} shows that the number of keys that must be stored in the stash
corresponds to the quantity $\bar{e}(G(S, h_0, h_1))$, defined as follows.  For a
connected component $H$ of $G(S, h_0, h_1)$, define the \emph{excess} $e(H) := \#\mbox{edges}(H) -
\#\mbox{nodes}(H)$, and define
$$\bar{e}(G(S, h_0, h_1)) := \sum_H \max(e(H),0),$$
where the sum is over all connected components $H$ of $G(S, h_0, h_1)$.  

Given a vertex in this random graph, recall that $C_{S, h_0, h_1}(v)$ denotes the number of edges
in the connected component that contains $v$, and let $B_v$ be the
number of edges in the component of $v$ that need to be removed
to make it acyclic ($B_v$ is also called the \emph{cyclotomic number} of the component containing $v$).  
Notice that $B_v=e(H)+1$, that is, $B_v$ is 1 more
than the number of keys from $v$'s component that need to be placed in the stash.

We use the following lemma from \cite{kmw-chs-09}.
\begin{lemma} \label{kmwlemma1} 
Let $|S|=n$ with $(1+\epsilon)n \leq m$ for some constant $\epsilon$, and let $G(S, h_0, h_1)$ be a cuckoo graph with vertex set $[m]\times[m]$. Then
$$\Pr(B_v \geq t~| C_{S, h_0, h_1}(v) = k) \leq \left ( \frac{3e^5k^3}{m} \right )^t.$$ \end{lemma}

\subsection{Proving Theorem \ref{thm:fastinsertsinner}}
\label{sec:proof1}
Throughout this section, we use the following notation, following that in \cite{arbitman2009amortized}. Let $\pi'$ be an $m^{1/3}$-bounded sequence of $p(n)$ operations.
Denote by $(x_1, \dots, x_{p(n)})$ the elements inserted by $\pi'$. For any integer $0< i \leq p(n)-m^{1/3}$,  let $S_i$  denote the set of elements that are
stored in the data structure just before the insertion of $x_i$, and let $\hat{S_i}$ denote $S_i$ together with the elements $\{x_i, x_{i+1}, \dots, x_{i+m^{1/3}}\}$, ignoring any deletions between time $i$ and time $i+m^{1/3}$. Since $\pi'$ is an 
$m^{1/3}$ bounded sequence, we have $|S_i| \leq m^{1/3}$ and $|\hat{S}_i| \leq 2m^{1/3}$ for all $i$. 

We need the following lemma, which states that all components of $G'$ are tiny.

\begin{lemma} \label{lem:G'compsize} Let $|S|=n'$, and let $G'(S, f_0, f_1)$ be any cuckoo graph with vertex set $[m']\times[m']$. Assume $n'/m' \leq 2m^{1/3}$. 
Then for any $k \leq m^{1/4}$, and for any node $v$, $\Pr[C_{S, f_0, f_1}(v)  \geq k] \leq m^{-\Omega(k)}$.\end{lemma}
\begin{proof} Lemma \ref{kmwlemma2} implies that
$$\Pr(C_{S, f_0, f_1}(v) \geq k) \leq \left (\frac{2k}{m^{1/3}} \right )^k \frac{1}{k!}.$$
For $k \leq m^{1/4}$ the probability that $C_{S, f_0, f_1}(v)  \geq k$ is at most
$m^{-\Omega(k)}$; applying a union-bound over all $v$, the probability
some connected component has size greater than $k$ is still at most $m^{-\Omega(k)}$. 
\end{proof}

\subsubsection{Showing \text{InnerStash} Stays Small}
\label{Cstashsec}
We begin by showing that \text{InnerStash} has size at most $s$ at all points in time with probability at least $1-1/m^{\Omega(s)}.$ 
This requires extending the analysis of prior works \cite{kmw-chs-09, kutzelnigg} to superconstant sized stashes, while
leveraging the sparsity of the inner cuckoo table $C$. 


\begin{lemma} \label{lem:innerstash} Let $|S|=n'$, and let $G'(S, f_0, f_1)$ be any cuckoo graph vertex set $[m'] \times [m']$. Assume $n' \leq 2m^{1/3}$ and $m' =m^{2/3}$. 
For any $s< m^{1/6}$, 
with probability $1 - 1/m^{\Omega(s)}$
the total number of vertices that reside in the stash of $G'(S, f_0, f_1)$ is at most $s$. 
\end{lemma}
The proof is entirely the same as that for the outer stash, which we prove in Lemma~\ref{smallprimarystashlemma}.
We defer the proof until then.  

We are in a position to show formally that the inner stash is small at  all points in time.

\begin{lemma} \label{finalinnerstashlemma} Let $\pi'$ be any $m^{1/3}$-bounded sequence of operations of length $p(n)$ on the inner cuckoo table $C$. For any $s<m^{1/12}$, with probability
$1-p(n)/m^{\Omega(s)}$ over the choice of hash functions $f_0, f_1$, \text{InnerStash} has size at most $s$ at all points in time.
\end{lemma}
\begin{proof}
Following the approach of \cite{arbitman2009amortized}, we define a \emph{good event} $\xi_1$ that ensures that \text{InnerStash} has size
at most $s$. Namely, let $\xi_1$ be the event that at all points $i$ in time, 
the number of vertices $v$ in the stash of $G'(\hat{S}_i, f_0, f_1)$ is at most $s$. By Lemma \ref{lem:innerstash}, 
and a union bound over all $p(n)$ operations in $\pi'$, $\xi_1$ occurs with probability at least $1-p(n)/m^{\Omega(s)}$.

We distinguish between the \emph{real} stash and the \emph{effective} stash. The real stash at any point $i$ in time is the set
of nodes that reside in the stash of $G'(S_i, f_0, f_1)$; the effective stash refers to the real stash, together with elements that used to be in the real stash
but have since had cycles removed from their components due to deletions, and have not yet been inserted in the inner cuckoo table. Let $E_i$
denote the set of items in the effective stash but not the real stash of $G'(S_i, f_0, f_1)$.

Clearly the event $\xi_1$ guarantees that at any point in time, the size of the real stash is at most $s$. To see that the size of the effective stash
never exceeds $s$, assume by way of induction that for $j \geq m^{1/3}$, the effective stash has size at most $s$ at all times less than $j$ (clearly event $\xi_1$ guarantees that this is true for all $j \leq m^{1/3}$ as a base case). In particular, the inductive hypothesis ensures that the effective stash has size at most $s$ at time $j-m^{1/3}$.
Observe that during the execution of operations $\{x_{j-m^{1/3}}, x_{j-m^{1/3}+1}, \dots, x_{j}\}$, we spend $\frac{m^{1/3}}{m^{1/6}}\geq s^2$ moves in total on the elements 
of the effective stash. By Lemma \ref{lem:G'compsize} and a union bound of all $i \leq p(n)$, we may assume all connected components of $G'(\hat{S}_j, f_0, f_1)$
have size at most $s/2$; this adds a failure probability of at most $p(n)/m^{\Omega(s)}$ to the result. Thus, by Lemma \ref{moveslemma} the insertion of any item $x$ in the effective stash never requires more than $s$ moves before it succeeds or causes $x$ to be returned to the back of the stash. 
It follows that by time $j$, we have spent at least $s$ moves on each of the items in $E_{j-m^{1/3}}$, and hence all elements in $E_{j-m^{1/3}}$ are inserted into the table by time $j$. Thus, at time $j$ the size of the effective stash is at most $s$, and this completes the induction.

\end{proof}

\subsubsection{Showing \text{InnerQueue} Stays Small}
\label{Cqueuesec}
The bulk of our analysis relies on the following technical lemma.

\begin{lemma} \label{lem:innerqueue} Let $|S|=n'$, and let $G'(S, f_0, f_1)$ be any cuckoo graph with vertex set $[m'] \times [m']$. Assume $n'/m' \leq 2m^{1/3}$. 
For any $0<s  \leq m^{1/6}$, 
with probability $1/m^{\Omega(\sqrt{s})}$
the total number of vertices  $v \in G'(S, f_0, f_1)$ with $C_{S, f_0, f_1}(v) > 8$ is at most $s$.
\end{lemma}
\begin{proof}
Lemma \ref{lem:G'compsize} implies that
the probability that all connected components have size at most $\sqrt{s}$ is at least $1-m^{-\Omega(\sqrt{s})}$. 
We assume for the remainder that there are no components 
of size greater than $\sqrt{s}$; this adds at most a probability $1/m^{\Omega(\sqrt{s})}$ of failure in
the statement of the lemma.

Pick a
special vertex $v_i$ in each component $C_i$ of size greater than 8 at random, and assign
to $v_i$ a ``weight'' equal to the number of edges in $C_i$. We show with probability at least $1-m^{-\Omega(\sqrt{s})}$ 
we cannot find a set of weight $s$; the lemma follows.

Since no component has size more than $\sqrt{s}$, we would need to find
at least $j = s/\sqrt{s} = \sqrt{s}$
vertices that are in components of size greater than 8 and are
the special vertex for that component.
We use the fact that for any fixed set of $j$ distinct vertices $v_1 \dots v_j$,
the probability (over both the choice of $G'$ and the choice 
of the special vertex for each component) that all $j$ vertices are the special vertices for their component is upper bounded
by $\Pr[C_{S, f_0, f_1}(v)  \geq 8]^j$. Indeed, write $E_i$ for the event $v_1, \dots v_{i-1}$ are all special. We may write $$\Pr[v_1 \dots v_j \text{ are all special } ] \leq \prod_{i=1}^j \Pr[v_i \text{ is special } |E_i].$$

To bound the right hand side, notice that  
$$\Pr[v_i \text{ is special } | E_i] \leq \Pr[v_i \text{ is in a different component from } v_1, \dots v_{i-1} \text{ and } C_{S, f_0, f_1}(v)  \geq 8 | E_i ]$$
$$ \leq \Pr[C_{S, f_0, f_1}(v) \geq 8],$$
where the last inequality holds because because the density of edges after taking out the earlier
(larger than expected) components is less than the density of edges a priori. 

Thus, taking a union bound over every set $\{v_1, \dots, v_j \}$ of $j = \sqrt{s}$ vertices,
the probability that we can find such a set is at most
$${m' \choose \sqrt{s}} \left( \Pr[v_1, \dots v_j \text{ are all special } ]\right)
\leq {m' \choose \sqrt{s}} \left( \Pr[C_{S, f_0, f_1}(v) > 8]^{\sqrt{s}}\right)$$
$$\leq {m' \choose \sqrt{s}} \left( \frac{8^8}{8!m^{8/3}} \right )^{\sqrt{s}}
\leq  m^{-\Omega(\sqrt{s})}.$$
where the second inequality follows by Lemma \ref{kmwlemma2}.
\end{proof}

We are in a position to show formally that the inner queue is small at  all points in time.

\begin{lemma} \label{finalinnerqueuelemma} Let $\pi'$ be any $m^{1/3}$-bounded sequence of operations of length $p(n)$ on the inner cuckoo table $C$. For any $s  \leq m^{1/6}$, with probability
$1-p(n)/m^{\Omega(\sqrt{s})}$ over the choice of hash functions $f_0, f_1$, \text{InnerQueue} has size at most $s$ at all points in time.
\end{lemma}
\begin{proof}
Following the approach of \cite{arbitman2009amortized}, we define a \emph{good event} $\xi_2$ that ensures that \text{InnerQueue} has size
at most $s$. Namely, let $\xi_2$ be the event that at all points $i$ in time, 
the number of vertices  $v \in G'$ with $C_{\hat{S}_i, f_0, f_1}(v) > 8$ is at most $s$. By Lemma \ref{lem:innerqueue}, 
and a union bound over all $p(n)$ operations in $\pi'$, $\xi_2$ occurs with probability at least $1-p(n)/m^{\Omega(\sqrt{s})}$.

We distinguish between the \emph{real} queue and the \emph{effective} queue. The real queue at any point $i$ in time is the set
of nodes $v$ such that $C_{S_i, f_0, f_1}(v) > 8$; the effective queue refers to the real queue, together with elements that used to be in the real queue
but have since had their components shrunk to size at most 8 via deletions, and have not yet been inserted in the inner cuckoo table. Let $E_i$
denote the set of items in the effective queue but not the real queue at time $i$.

Clearly the event $\xi_2$ guarantees that at any point in time, the size of the real queue is at most $s$. To see that in fact the size of the effective queue
never exceeds $s$, assume by way of induction that for $j \geq m^{1/3}$, the effective queue has size at most $s$ at all times less than $j$ (clearly event $\xi_2$ guarantees that this is true for all $j \leq m^{1/3}$ as a base case). In particular, the inductive hypothesis ensures that the effective queue has size at most $s$ at time $j-m^{1/3}$.
Observe that during the execution of operations $\{x_{j-m^{1/3}}, x_{j-m^{1/3}+1}, \dots, x_{j}\}$, we spend $16\frac{m^{1/3}}{m^{1/6}}\geq 16s$ moves on the elements 
of the effective queue, with at least $16$ moves devoted to each element. 
By definition of the real queue, combined with Lemma \ref{moveslemma}, at most 16 moves are required to insert each element of $E_{j-m^{1/3}}$, and 
thus all elements in $E_{j-m^{1/3}}$ are inserted into the table by time $j$. Thus, at time $j$ the size of the effective queue is at most $s$, and this completes the induction.
\end{proof}


\subsubsection{Putting it together}
By design, every insert into $C$ terminates in $O(1)$ steps. Combining Lemmata \ref{finalinnerstashlemma} and \ref{finalinnerqueuelemma}, with probability $1-p(n)/m^{-\Omega(\sqrt{s})}$, both \text{InnerQueue} and \text{InnerStash} have size at most $s/2$ at all times. Since $L$ only contains 
items from the inner queue and inner stash, it follows that $L$ never contains more than $s$ items. This proves Theorem \ref{thm:fastinsertsinner}.



\subsection{Proving Theorem \ref{thm:Qsmall}}
\label{sec:proof2}
Throughout this section, we use the following notation, following that in Section \ref{sec:proof1}. Let $\pi$ be an $n$-bounded sequence of $p(n)$ operations.
Denote by $(x_1, \dots, x_{p(n)})$ the elements inserted by $\pi$. For any integer $0< i \leq p(n)$,  let $S_i$  denote the set of elements that are
stored in the data structure just before the insertion of $x_i$, let $\hat{S_i}$ denote $S_i$ together with the elements $\{x_i, x_{i+1}, \dots, x_{i+\log^6 n}\}$, ignoring any deletions between time $i$ and time $i+\log^6 n$, and let $\bar{S_i}$ denote $S_i$ together with the elements $\{x_i, x_{i+1}, \dots, x_{i+m^{1/2}}\}$, ignoring any deletions between time $i$ and time $i+m^{1/2}$ (treating any operations past time $p(n)$ as empty). Since $\pi$ is an 
$n$ bounded sequence, we have $|S_i| \leq n$, $|\hat{S}_i| \leq n + \log^6 n$, and $|\bar{S}_i| \leq n + m^{1/2}$ for all $i$. 

In proving Theorem \ref{thm:Qsmall} it suffices to show that neither the stash nor the queue of the primary structure will grow too large.
We will use the following lemma.

\begin{lemma} \label{lem:propertytwo} Let $|S|=n$, and let $G(S, h_0, h_1)$ be a cuckoo graph with vertex set $V=[m] \times [m]$, with $m\geq (1+\epsilon)n$ for some constant $\epsilon >0$.  
There exists a constant $c_2>0$ such that with probability $1-1/n^{\Omega(\log n)}$ over the choice of $h_0$ and $h_1$, all components of $G(S, h_0, h_1)$ are of size at most $c_2\log^2 n$.
\end{lemma} 
\begin{proof}
It is a standard calculation that for any node $v$ in the cuckoo graph, $\Pr[C_{S, h_0, h_1}(v) \geq k] \leq \beta^k$ for some
constant $\beta \in (0, 1)$ (see e.g. \cite[Lemma 2.4]{kmw-chs-09}). The conclusion follows by setting $k=O(\log^2 n)$, and applying the union bound 
over all vertices.
\end{proof}

\subsubsection{Showing \text{OuterStash} Stays Small}
\label{Tstashsec}
\begin{lemma} \label{smallprimarystashlemma}
 Let $|S|=n$, and let $G(S, h_0, h_1)$ be any cuckoo graph with vertex set $V=[m] \times [m]$, where $(1+\epsilon)n \leq m$ for some constant
 $\epsilon > 0$.
For any $s$ such that $s \leq m^{1/6}$, with probability $1/m^{\Omega(s)}$
the total number of vertices that reside in the stash of $G(S, h_0, h_1)$ is at most $s$. 
\end{lemma}
The choice of $s \leq m^{1/6}$ is fairly arbitrary but convenient and sufficient for our purposes.
Notice Lemma~\ref{smallprimarystashlemma} readily implies Lemma~\ref{lem:innerstash}.

\begin{proof}
Here we follow the work of \cite{GoodrichMitzenmacherICALParXivVersion}, which considers the
analysis of super-constant sized stashes, extending the work of \cite{kmw-chs-09}.


The starting point is Lemmata~\ref{kmwlemma2} and~\ref{kmwlemma1} above.  Specifically,
in \cite{GoodrichMitzenmacherICALParXivVersion} it is shown that these lemmata imply that
$$\Pr(B_v \geq t) \leq \sum_{k=1}^\infty \min \left( \left( \frac{3e^5k^3}{m} \right)^t,1 \right) \beta^k$$ for
some constant $\beta$.  In particular, we can concern ourselves with values of $k$ that are 
$O(m^{1/5})$, since the summation over $\beta^k$ terms for larger values of $k$ is dominated
by $2^{-\Omega(m^{1/5})} = m^{-\Omega(m^{1/6})}$.

It follows that $\Pr(B_v \geq t)$ is at most $\max\left(m^{-\Omega(t)},m^{-\Omega(m^{1/6})}\right)$.
We therefore claim that $\Pr(B_v \geq j+1) \leq m^{-1-\alpha j}$ for some constant $\alpha$
for $j \leq m^{1/6}$.  

Now, following the derivation of Theorem 2.2 of \cite{kmw-chs-09}, we have the probability that
the stash exceeds size $s$ is given by the probability that $2m$ independently 
chosen components have an excess of more than $s$ edges, which can be bounded as:
\begin{eqnarray*}
\Pr(\bar{e}(g) \geq s)& \leq & \sum_{k=1}^{2m} {2m \choose k} s^{k} m^{-\alpha s- k}\\
& \leq & \sum_{k=1}^{2m}  m^{-\alpha s} \left( \frac{2es}{k} \right)^k \\
& \leq & (2m) m^{-\alpha s} e^{2s} \\
&  = & m^{-\Omega(s)}.
\end{eqnarray*}
Here the second to last line follows from a straightforward optimiziation to find the the maximum of
$(x/k)^k$ (which occurs at $k=x/e$). 
\end{proof}

We now show formally that the outer stash is small at  all points in time.

\begin{lemma} \label{finalouterstashlemma} Let $\pi$ be any $n$-bounded sequence of operations of length $p(n)$.  For $s \leq m^{1/5}$, with probability
$1-p(n)/m^{\Omega(s)}$ over the choice of hash functions $h_0, h_1$, \text{OuterStash} has size at most $s$ at all points in time.
\end{lemma}
\begin{proof}
We define a \emph{good event} $\xi_3$ that ensures that \text{OuterStash} has size
at most $s$. Namely, let $\xi_3$ be the event that at all points $i$ in time, 
the number of vertices in the stash of $G(\bar{S}_i, h_0, h_1)$ is at most $s$, and additionally all connected components 
of $G(\bar{S}_i, h_0, h_1)$ have size at most $\log^2 n$. By Lemmata  \ref{lem:propertytwo} and \ref{smallprimarystashlemma}, 
as well as a union bound over all $p(n)$ operations in $\pi$, $\xi_3$ occurs with probability at least $1-p(n)/m^{\Omega(s)}$.
As in \cite{arbitman2009amortized}, a minor technical point in applying Lemma \ref{smallprimarystashlemma} is that $\bar{S}_i$ is $n+m^{1/2}$-bounded, not $n$-bounded.
But we can handle this by applying Lemma \ref{smallprimarystashlemma} with $\epsilon' = \epsilon/2$, since for large enough $m$, 
$(1+\epsilon/2)(n+m^{1/2})\leq (1+\epsilon)n\leq m$.

We distinguish between the \emph{real} stash and the \emph{effective} stash. The real stash at any point $i$ in time is the set
of nodes $v \in V$ such that reside in the stash of $G(S_i, h_0, h_1)$; the effective stash refers to the real stash, together with elements that used to be in the real stash but have since had cycles removed from their components due to deletions, and have not yet been inserted in the outer cuckoo table. Let $E_i$
denote the set of items in the effective stash but not the real stash of $G(S_i, h_0, h_1)$.

Clearly the event $\xi_3$ guarantees that at any point in time, the size of the real stash is at most $s$. To see that the size of the effective stash
never exceeds $s$,  assume by way of induction that for $j \geq m^{1/2}$, the effective stash has size at most $s$ at all times less than $j$ (clearly event $\xi_3$ guarantees that this is true for all $j \leq m^{1/2}$ as a base case). In particular, the inductive hypothesis ensures that the effective stash has size at most $s$ at time $j-m^{1/2}$.
Observe that during the execution of operations $\{x_{j-m^{1/2}}, x_{j-m^{1/2} +1}, \dots, x_{j}\}$, we spend at least $m^{1/2}/m^{1/4}=m^{1/4}$ moves in total on the elements 
of the effective stash. Since all connected components have size at most $\log^2 n$, Lemma \ref{moveslemma} implies the insertion of any item $x$ in the effective stash never requires more than $2 \log^2 n$ moves before it succeeds or causes $x$ to be returned to the back of the stash. Thus,
all elements in the effective stash at time $j-m^{1/2}$ require at most $2m^{1/5}\log^2n \leq m^{1/4}$ operations in total to process, and it
follows that all elements in $E_{j-m^{1/2}}$ are inserted into the table by time $j$. Thus, at time $j$ the size of the effective stash is at most $s$, and this completes the induction.
\end{proof}

\subsubsection{Showing \text{OuterQueue} Stays Small} 
\label{Tqueuesec}
Let $|S|=n$, and let $G(S, h_0, h_1)$ be any cuckoo graph with vertex set $[m] \times [m]$. 
It is well-known that, for any node $v$, there is significant probability that an insertion of $v$ into $G$ takes $\Omega(\log n)$ time.
But one might hope that for a sufficiently large set of distinct vertices $\{v_1, \dots, v_N\}$, the 
\emph{average} size of the connected components of the $v_i$'s is constant with overwhelming probability over choice of $G$, 
and thus any sequence of $N$ insertions
will take $O(N)$ time in total.
Indeed, Lemma 4.4 of Arbitman {\it et al}. \cite{arbitman2009amortized} establishes that, for any distinct vertices $\{v_1, \dots, v_N\}$ with $N \leq \log n$, 
$\sum_{i=1}^N C_{S, h_0, h_1}(v_i) = O(N)$ with probability $1-2^{-\Omega(N)}$ over the choice of $h_0$ and $h_1$. 
Roughly speaking, Arbitman {\it et al}. use this result to conclude that a logarithmic sized queue suffices for deamortizing cuckoo hashing (where all guarantees hold with
high probability), since any sequence of
$\log n$ insertions can be processed in $O(\log n)$ time steps.

Our goal is to achieve guarantees which hold with overwhelming probability.
To achieve this, we use the fact that we can afford to keep a queue of super-logarithmic size without affecting the asymptotic
space usage of our algorithm.
We show that any sequence of, say, $N=\log^6 n$ operations can be cleared from the queue in $O(N)$ time with probability
$1-1/n^{\omega(1)}$. It follows that with overwhelming probability the queue does not overflow. 
Unfortunately, the techniques of \cite{arbitman2009amortized} do not generalize to values of $N$ larger than $O(\log n)$, so we generalize
their result using different methods. 

Intuitively, one should picture the random process we wish to analyze as a 
standard queueing process, where the time to handle each job is a random
variable.  In our case, the random variable for a job -- which is a
key $k$ to be placed -- is the time to find a spot for $k$ in the cuckoo table, which is proportional to the size of the
connected component in which $k$ is placed in the cuckoo graph  by Lemma \ref{moveslemma}. This random variable is known to be constant on
average and have exponentially decreasing tails, so if the job times were
independent, this would be a normal queue (more specifically, a Galton-Watson process), and the bounds would follow from standard analyses.

Unfortunately, the job times in our setting are not independent. 
Roughly speaking, our analysis proceeds by showing that with overwhelming probability on a given instance
of the cuckoo graph, the expected value of the size of the connected component of a randomly chosen vertex is close to its expectation if the 
graph was chosen randomly.

Our main technical tool will be the following lemma.
\begin{lemma} \label{queuelemma} Let $|S|=n$, and let $G(S, h_0, h_1)$ be a cuckoo graph with vertex set $V=[m] \times [m]$. Let $\{v_1, \dots, v_N\}$ be a set of $N>\log^6(n)$ vertices chosen uniformly at random (with replacement) from the cuckoo graph $G(S, h_0, h_1)$.
There is a constant $c$ such that with probability
$1-n^{-\Omega(\log n)}$ (over both the choice of the $v_i$'s and the generation of the cuckoo graph $G$), $\sum_{i=1}^N C_{S, h_0, h_1}(v_i) \leq cN$.
\end{lemma}

\begin{proof}
Let $X_i$ be the number of vertices in $G(S, h_0, h_1)$ in components of size $i$, and let $\mu_i = \E[X_i]$ be the expected number of vertices
in components of size $i$ in a random cuckoo graph, and let $\mu = \frac{1}{2m} \sum_{i=1}^{2m} i \mu_i$ be the expected size of the connected component of a random node in a random cuckoo graph. By standard calculations (see e.g.\cite[Lemma 2.4]{kmw-chs-09}), there is a constant $\beta \in (0, 1)$ such that 
for any $v$, $\E[C_{S, h_0, h_1}(v)] \leq \sum_{k=1}^{\infty} \beta^k = O(1)$, where the expectation is taken over choice of $h_0$ and $h_1$, and hence $\mu=O(1)$.

For \emph{fixed} hash functions $h_0$ and $h_1$, we can write \begin{equation} \label{eq:cuckoo} \E_{v \in V} C_{S, h_0, h_1}(v) = \frac{1}{2m} \sum_{v \in V} C_{S, h_0, h_1}(v) = \frac{1}{2m}\sum_{i=1}^{2m} i X_i. \end{equation} 
Our goal is to show that with overwhelming probability over the choice of $h_0$ and $h_1$, the right hand side of Equation \ref{eq:cuckoo} is close to $\mu$. 
We will do this by showing that for small enough $i$, $X_i$ is tightly concentrated around $\mu_i$, and that larger $i$ do not contribute significantly
to the sum. 
 
\begin{lemma} \label{lemma1} The following properties both hold.
\begin{enumerate}
\item Suppose $\mu_i \geq n^{2/3}$.
Then $\Pr(|X_i - \mu_i| \geq n^{2/3}) = 2e^{-\widetilde{\Omega}(n^{1/3})}.$
(The $\widetilde{\Omega}$ notation hides factors polylogarithmic in $n$.)
\item Let $i^*$ be the smallest value of $i$ such that $\mu_i <
n^{2/3}$, and let $X_*$ be the number of vertices in components of
size at least $i^*$.  Then $\Pr[X_* \geq \gamma n^{2/3}] \leq 1/n^{\Omega(\log n)}$ for some constant $\gamma$.
\end{enumerate}
\end{lemma}
\begin{proof}

\begin{enumerate} 
\item This follows from a standard application of Azuma's inequality,
applied to the edge exposure martingale that reveal the edges of the cuckoo graph one
at a time. More specifically, we reveal the edges of $G$ one at a time in an arbitrary order, say
$e_1, \dots e_{n}$ and let $Z_j=E[X_j|e_1, \dots, e_{j-1}] $. Then $Z_j$ is a martingale, and changing a single edge can 
only change the number of components of size $j$ by a constant (specifically, two), and hence $|Z_j - Z_{j-1}| \leq 2j$ for all $j$. Thus,  
by Azuma's inequality 
$$\Pr(|X_i - \E[X_i]| \geq 2i \lambda \sqrt{n}) = \Pr(|Z_{n} - Z_0| \geq 2i\lambda \sqrt{n}) \leq 2e^{-\lambda^2/2}.$$ 
Setting $\lambda=\frac{n^{2/3}}{2i\sqrt{n}}=\Omega(n^{1/6}/i)$, we see 

$$\Pr(|X_i - \E[X_i]| \geq n^{2/3}) \leq e^{-\Omega(n^{1/3}/i^2)}.$$
By a standard calculation, $\E[X_i] \geq n^{2/3}$ implies $i \leq c_1 \log n$ for some constant $c_1$, and the theorem follows.

\item Since $\Pr[C_{S, h_0, h_1}(v) \geq k] \leq \beta^k$ for some constant $\beta \in (0, 1)$, it follows easily that 
$E[X_*] \leq \sum_{i^*}^\infty n \beta^i$, where $n\beta^{i^*} \leq n^{2/3}$.  It follows that $E[X_*] = O(n^{2/3})$.  

In order to get concentration of $X_*$ about its mean, we use a slight modification of the edge exposure martingale, which will essentially
allow us to assume that all connected components have size $O(\log^2 n)$ when attempting to bound the differences between martingale steps, which happens with very high probability by Lemma \ref{lem:propertytwo}. This technique is formalized for example in Theorem 3.7 of \cite{McDiarmid-98-concentration}.  

Let $Q$ be the event that all connected components are of size at most $\log^2 n$.
We reveal the edges of $G$ one at a time in an arbitrary order, say
$e_1, \dots e_n$ and let $Z_j=E[X_{*}|e_1, \dots, e_{j-1}] $ if the
edges $e_1, \dots e_{j-1}$ do not include a component of 
size greater than $\log^2 n$, and $Z_j=Z_{j-1}$ if 
the
edges $e_1, \dots e_{j-1}$ do include a component of 
size greater than $\log^2 n$.
Then $Z_j$ is a martingale, and since changing a single edge can only change
the number of components of size $i$ by at most two,
we see that $|Z_i - Z_{i-1}| \leq 2c\log^2 n$.  
Now as $Z_n$ will equal $X_*$ except in the case where event $Q$ does not hold, 
we can apply Azuma's inequality to the above martingale to conclude that
$$\Pr(|X_{*} - \E[X_{*}]| \geq 2c \lambda
  \sqrt{n} \log n) \leq 2e^{-\lambda^2/2} +Pr(\neg Q).$$ Setting
$\lambda=\frac{n^{1/3}}{c\log^2(n) \sqrt{n}}$, we
obtain $$\Pr(|X_{*} - \E[X_{*}]| \geq n^{1/3}) \leq
e^{-\widetilde{\Omega}(n^{2/3})} + \frac{1}{n^{\Omega(\log n)}}.$$
Noting that $\E[X_*] = O(n^{2/3})$,  we conclude $$\Pr(X_{*} \geq \gamma n^{1/3}) \leq \frac{1}{n^{\Omega(\log n)}}$$
for some constant $\gamma$.  
\end{enumerate}
\end{proof}

Properties 1 and 2 of Lemma \ref{lemma1} together imply that with very high probability over choice of $h_0$ and $h_1$, 
$$\sum_{i=1}^{2m} i X_i = \sum_{i=1}^{i^*-1} \left(\mu_i \pm n^{2/3} \right) + O(n^{2/3}) \leq \sum_{i=1}^{2m} 2 \mu_i  + O(n^{2/3} \log^2 n).$$
Combining the above with Equation \ref{eq:cuckoo},  with overwhelming probability over choice of $h_0$ and $h_1$ it holds that
 $$\E_{v \in V} C_{S, h_0, h_1}(v) = \frac{1}{2m}\sum_{i=1}^{2m} i X_i \leq 2 \mu + o(1)=O(1).$$
Thus, we have shown that with very high probability
over the choice of $G$, there is a constant $c_3$ such that $\E_{v \in V}[C_{S, h_0, h_1}(v)] \leq c_3$.

Our last step in proving Lemma \ref{queuelemma} is to show that if we choose a set of vertices at random, the sum of the
component sizes is concentrated around its mean. Indeed, by applying a similar argument as in the proof of Lemma \ref{lemma1} Property Two,
we can assume for all $v \in V$, $C_{S, h_0, h_1}(v) \leq c_2 \log^2 n$, as long as we add an additional term of $\frac{1}{n^{\Omega(\log n)}}$ in the
bound on the probability obtained using Azuma's inequality.
This yields
$$\Pr(|\sum_{i=1}^S C_{S, h_0, h_1}(v_i)- c_3| \geq \lambda \sqrt{S} c_2 \log^2 n) \leq 2 e^{-\lambda^2/2} + \frac{1}{n^{\Omega(\log n)}}.$$
Setting $\lambda = \sqrt{S}/\log^2 n$ yields 
$$Pr(|\sum_{i=1}^S C_{S, h_0, h_1}(v_i) - c_3| \geq c_2) \leq e^{-\Omega(S/\log^4 n)}+ \frac{1}{n^{\Omega(\log n)}}.$$
The conclusion follows, with $c=c_3 + c_2$.
\end{proof}

Notice that for any set of \emph{distinct} items $\{x_1, \dots, x_N\}$ to be inserted, the sets $\{h_0(x_1), \dots, h_0(x_N)\}$ and $\{h_1(x_1), \dots, h_1(x_N)\}$  are uniformly distributed set of vertices in $G(S, h_0, h_1)$. 
Thus, Lemma \ref{queuelemma} ensures that for distinct items $\{x_1, \dots, x_N\}$ with $N \geq \log^6 n$, and for any set $S$ of size $n$, with probability
at least $1-1/n^{\log n}$ over the choice of $h_0$ and $h_1$, $\sum_{i=1}^N C_{S, h_0, h_1}(v_i) \leq cN$. 

With this in hand, we are ready to show that with overwhelming probability \text{OuterQueue} does not exceed size $\log^6 n$ over any sequence of $\poly(n)$ operations.

\begin{lemma} \label{smallprimaryqueue}
With probability $1-n^{-\Omega(\log n)}$, the queue of the primary structure has size at most $\log^6 n$ at all times. 
\end{lemma}
\begin{proof}
We define a \emph{good event} $\xi_4$ that ensures that \text{OuterQueue} has size
at most $\log^6 n$. Namely, let $\xi_4$ be the event that for all times $\log^{6} n\leq j \leq p(n)$, 
it holds that $\sum_{i=j-log^6 m}^jC_{\hat{S}_j, f_0, f_1}(x_i) \leq c\log^6 n$. By Lemma \ref{queuelemma}, 
and a union bound over all $p(n)$ operations in $\pi$, $\xi_4$ occurs with probability at least $1-1/n^{\log n}$.

For $j > \log^6 n$, suppose by induction there are at most $\log^6 n$ items in the effective queue at all times less than $j$ (as a base case, this is clearly true for all $j \leq \log^6 n$). In particular, this holds at time $j-\log^6 n$.
We
can assume all items in the effective queue at time $j-\log^6 n$ are distinct because we process deletions immediately. 
Since all the (at most) $\log^6 n$ items in the effective queue are distinct, event $\xi_4$ guarantees that all of these items can be cleared from the queue in $c \log^6 n$ steps.  Setting the number of steps expended on elements of the queue every operation to $\alpha=c$, all
(at most) $\log^6 n$ items in the effective queue at time $j-\log^6 n$ will be cleared from the queue 
before time $j$. Thus, at time $j$, the queue contains at most $\log^6 n$ items, and this completes the induction.
\end{proof}

$Q$ only contains items from \text{OuterStash} and \text{OuterQueue}, and we have shown both deques contain
$\log^6 m$ items with overwhelming probability. Theorem \ref{thm:Qsmall} follows.

\subsection{Putting it All Together}
 For any constant $\epsilon > 0$, our nested cuckoo construction uses $(2+\epsilon)n$ words for the outer cuckoo table, $O(m^{2/3})$ for the inner structure $Q$, and (with overwhelming probability) $O(\log^2 m)$ words for the cycle-detection mechanisms. The latter two space costs are dominated by the first, and so 
 our total space usage is $(2+\epsilon)n$ words in total for constant $\epsilon >0$. 
 
 We derive our final theoretical guarantees on the running time of each operation for both of our constructions.
 
\textbf{Construction One:} Inserts take $O(1)$ time by design. Lookups and deletions require examining $T_0[h_0(k)]$,
$T_1[h_1(k)]$,
$R_0[f_0(k)]$,
and $R_1[f_1(k)]$, and 
performing a lookup in $L$, which in Construction 1 potentially requires examining all elements in $L$. 
 Theorems \ref{thm:fastinsertsinner} and \ref{thm:Qsmall} together imply that for $0 < s < \log^2 n$, with probability at least $1-1/m^{\Omega(\sqrt{s})}$,
 $L$ contains only $s$ items.
 Thus, lookups and removals take time $s$ with probability at least $1-1/m^{\Omega(\sqrt{s})}$, even though $s$
 is not a tuning parameter of our construction. 
 
 To clarify, the hidden constant in the exponent is \emph{fixed, i.e. independent of all parameters}. Thus, 
 for any constant $c$, there is some larger constant $s$ such that the probability lookups and removals take time more than $s$ is bounded by $1-1/m^c$.
 We also remark that it is straightforward to extend our analysis to all $s \leq m^{1/12}$, but for clarity
 we have not presented our results in this generality.
 
\textbf{Construction Two}: Again, inserts take $O(1)$ time by design. As in Construction One,  lookups and deletions require examining $T_0[h_0(k)]$,
$T_1[h_1(k)]$,
$R_0[f_0(k)]$,
and $R_1[f_1(k)]$, and 
performing a lookup in $L$; assuming $L$ has size $s=O(\log^{1/2} n)$, such a lookup can be performed in constant time using our atomic stash.
 Theorems \ref{thm:fastinsertsinner} and \ref{thm:Qsmall} therefore imply that lookups and removes take $O(1)$ time with probability $1-1/m^{\Omega(\log^{1/4}{n})}$.
 
\subsection{Extensions}
\label{sec:extensions}
\subsubsection{$\poly\log(n)$-wise Independent Hash Functions}
We remark that an argument of Arbitman {\it et al}. \cite{arbitman2009amortized} implies almost without modification that in Construction Two, insertions, deletions,
and lookups take $O(1)$ time with overwhelming probability even if the hash functions $h_0$, $h_1$, $f_0$, and $f_1$
are chosen from $\poly \log(n)$-wise independent hash families. For completeness, we reproduce this argument in our context.

In the analysis above, the only places we used the independence of our hash functions were in Lemmata \ref{lem:innerstash}, \ref{lem:innerqueue}, \ref{smallprimarystashlemma}, and \ref{queuelemma} above. These lemmata allowed us to define four events that occur with high or overwhelming probability, whose occurrence guarantee that the our time bounds hold. Specifically, for any fixed $s$, our time bounds for Construction Two hold if none of the following ``bad'' events  occur:

\begin{enumerate}
\item Event 1: There exists a set $\{v_1, \dots, v_{N}\}$ of $N=O(\log^6 n)$ vertices in the outer cuckoo graph, such that  $\sum_{i=1}^N C_{S, h_0, h_1}(v_i) > cN$ (this is the complement of event $\xi_4$ from Section \ref{Tqueuesec}).
\item Event 2: There exists a set  of at most $2m$ vertices in the outer cuckoo graph, such that the number
of stashed elements from the set exceeds $O(\log^6 n)$ (this is the complement of event $\xi_3$ from Section \ref{Tstashsec}).
\item Event 3: There exists a set of at least $O(\log^{1/2} n)$ vertices in the inner cuckoo graph, all of whose connected components at some point in time 
have size greater than 8 (this is the complement of the event $\xi_2$ defined in Section \ref{Cqueuesec}).
\item Event 4: There exists a set of vertices in the cuckoo graph, such that the number
of stashed elements from the set at some point in time exceeds $O(\log^{1/2} n)$ (this is the complement of the
event $\xi_1$ defined in Section \ref{Cstashsec}).
\end{enumerate}

Lemmata  \ref{lem:innerstash}, \ref{lem:innerqueue}, \ref{smallprimarystashlemma}, and \ref{queuelemma}  ensure that if $h_0, h_1, f_0$, and $f_1$ are fully random, none of the four events occur with probability at least $1-1/n^{\Omega(\log^{1/4}(n))}$. In order to show that the conclusion holds even if the hash functions are $\poly \log(n)$-wise independent, we apply a recent result of Braverman \cite{braverman09} stating that polylogarithmic independence fools constant-depth boolean circuits. 

\begin{theorem} \label{bravthm} (\cite{braverman09}) Let $s \geq \log m$ be any parameter. Let $F$ be a boolean function computed
by a circuit of depth $d$ and size $m$. Let $\mu$ be an $r$-independent distribution where
$$r \geq 3 \cdot 60^{d+3} \cdot (\log m)^{(d+1)(d+3)} \cdot s^{d(d+3)},$$
then
$|\E_\mu[F] - \E[F]| < 0.82s \cdot 15m$,
\end{theorem}

Theorem \ref{bravthm} implies that if we can develop constant-depth boolean circuits of quasi-polynomial size that recognize Events 1-4 above, then the 
probability any of the events occur under polylogarithmic independent hash functions will be very close to the probability the events occur under fully random hash functions. The circuits that recognize our events are similar to those used in Arbitman {\it et al}. \cite{arbitman2009amortized}; the input wires to the first two circuits contain the values $h_0(x_1), h_1(x_1), . . . , h_0(x_{n}), h_1(x_{n})$ (where
the $x_i$Õs represent the elements inserted into the outer cuckoo table), while the input wires to the second two circuits 
contain the values $f_0(x_1), f_1(x_1), . . . , f_0(x_{j}), f_1(x_{j})$, where
$j$ is the number of items in the inner cuckoo table.

\begin{enumerate}

\item \textbf{Identifying Event 1}: Just as in \cite{arbitman2009amortized}, this event occurs if and only if the graph contains at least one forest from
a specific set of forests of the bipartite graph on $[m]\times[m$], where $m=(1+\epsilon)n$. We denote this set of
forests by $\mathcal{F}_n$, and observe that $\mathcal{F}_n$ is a subset of all forests with at most $cN = O(\log^6(n))$ vertices,
which implies that $|\mathcal{F}_n| = n^{\poly \log(n)}$. Therefore, the event can be identified by a constant-depth
circuit of size $n^{\poly \log(n)}$ that simply enumerates all forests inn $\mathcal{F}_n$, and for every such forest checks whether it
exists in the graph.

\item \textbf{Identifying Event 2}: A constant-depth circuit identifying this event enumerates over all $S \subseteq [m]\times[m]$ of size $\log^{1/2}(n)$
and checks whether all of elements of $S$ are stashed. As in \cite{arbitman2009amortized}, a minor complication is that we must define a canonical set of stashed elements for each set $S$; this is only for simplifying the analysis, and does not require modifying our actual construction. Our circuit checks whether all elements of $S$
are stashed by, for each $x \in S$, enumerating over all connected components in which edge $(h_0(x), h_1(x))$ is stashed according to the canonical set of stashed items for $S$ and checking if the component exists in the graph. We may assume Event 1 does not occur, and thus we need only iterate over components of $O(\log^6(n))$ vertices. The circuit thus has $O(n^{\poly\log(n)})$ size.

\item \textbf{Identifying Events 3 and 4:} We may assume Events 1 and 2 do not occur. Then there are at most $O(\log^6 n)$ edges in the inner cuckoo table,
so a constant depth circuit of quasipolynomial size simply enumerates over all possible edge sets $E'$ of size $O(\log^6 n)$ satisfying Event 3 or Event 4
and checks if $E'$ equals the input to the circuit.
\end{enumerate}

Thus, we can set $s=\poly\log(n)$ and $m=\poly\log(n)$ in the statement of Theorem \ref{bravthm} to conclude that, even if we use hash functions from a  $\poly \log(n)$-wise independent family of functions, Events 1-4 still only occur with negligible probability. We remark that similar arguments demonstrate
that when $\poly\log(n)$-wise independent hash functions are used, all operations under Construction 1 still take $O(s)$ time with probability $1-1/m^{\Omega(\sqrt{s})}$ when $s=\poly\log(n)$ (the amount of independence required depends on $s$). 

\subsubsection{Sufficiently Independent Hash Functions Evaluated in $O(1)$ time}
\label{sec:weakhash}
Unfortunately, all known constructions of $\poly \log(n)$-wise independent hash families that can be evaluated in the RAM model while maintaining 
our $O(n)$ space bound come with important caveats. The classic construction of $k$-wise independent hash functions due to Carter and Wegman based on degree-$k$ polynomials over finite fields requires time $O(k)$ to evaluate; ideally we would like $O(1)$ evaluation time to maintain our time bounds.
The work of Siegel \cite{siegel04} is particularly relevant; for polynomial-sized universes, he proves the existence of a family of $n^{\epsilon}$-wise independent hash functions for $\epsilon>0$ (which is super-logarithmic), which can be
evaluated in $O(1)$ time using look up tables of size $n^{\delta}$ for some $\delta<1$. However, his construction is non-uniform in that he relies on the existence
of certain expanders for which we do not possess explicit constructions. 
Subsequent works that improve and/or simplify \cite{siegel04} (e.g. \cite{dw03, pp08, dr09}) all possess polynomial probabilities of failure, which render them
unsuitable when seeking guarantees that hold with overwhelming probability.
The development of uniformly computable hash families which can be evaluated in $O(1)$ time using $o(n)$ words of memory remains
an important open question.

\subsubsection{Achieving Loads Close to One}
\label{sec:betterloads}
We remark that we can achieve $O(1)$ worst-case operations with overwhelming probability using $(1+\epsilon)n$ words of memory for any constant $\epsilon > 0$
by substituting our fully de-amortized nested cuckoo hash tables in for the ``backyard'' cuckoo table of Arbitman {\it et al}. \cite{ans-bchcw-10}. The construction of \cite{ans-bchcw-10} uses a main table consisting of $(1+\epsilon/2) d$ buckets, each of size $d$ for some constant $d$, and 
uses a de-amortized cuckoo table as a ``backyard'' to handle elements from overflowing buckets. They show that for constant $\epsilon > 0$,
with overwhelming probability the backyard cuckoo table must only store a small constant fraction of the elements.
Note that $n^{\alpha}$-wise independent hash functions for some $\alpha>0$ are required to map items to buckets in the main table, and therefore the technique cuts our space usage by a factor of about 2, but increases the amount of independence we need to assume in our hash functions for theoretical guarantees to hold.

\vspace{-4mm}
\section{Cache-Oblivious Multimaps} \label{sec:multimap}
In this section, we describe our cache-oblivious implementation of
the multimap ADT.
To illustrate the issues that arise in the construction, we first give a simple implementation for a RAM, and then give an improved
(cache-oblivious) construction for the external memory model.  Specifically, we describe an amortized cache-oblivious solution
and then we describe how to de-amortize this solution.

In the implementation for the RAM model, we maintain two nested cuckoo hash tables, as described in Section \ref{sec:cuckoo}. The first
table enables fast $\text{containsItem}(k, v)$ operations; this table stores all the $(k,v)$ pairs using each
entire key-value pair as the key, and the value associated with $(k, v)$ is a pointer to $v$'s entry in a linked list $L(k)$ containing all values associated with $k$ in the multimap. The second table ensures fast $\text{containsKey}(k)$, $\text{getAll}(k)$, and $\text{removeAll}(k)$ operations:
this table stores all the unique keys $k$, as well as a pointer to the head of $L(k)$. 

\medskip
\noindent \textbf{Operations in the RAM implementation.}

\begin{enumerate} 
\item $\text{containsKey}(k)$: We perform a lookup for $k$ in Table~2.
\item $\text{containsItem}(k,v)$: We perform a lookup for $(k, v)$ in Table~1.
\item $\text{add}(k,v)$: We add $(k, v)$ to Table~1 using the insertion procedure of Section \ref{sec:cuckoo}. We perform
a lookup for $k$ in Table~2, and if $k$ is not found we add $k$ to Table 2. We then insert $v$ as the head of the linked list corresponding to Table 2. 
\item $\text{remove}(k,v)$: We remove $(k, v)$ from Table~1, and remove $v$ from the linked list $L(k)$; if $v$ was the head of $L(k)$, we also perform a lookup for
$k$ in Table~2 and update the pointer for $k$ to point to the new head of $L(k)$ (if $L(k)$ is now empty, we remove $k$ from Table~2.)
\item $\text{getAll}(k)$: We perform a lookup for $k$ in Table~2 and return the pointer to the head of $L(k)$.
\item $\text{removeAll}(k)$: We remove $k$ from Table~2. In order to achieve unamortized $O(1)$ I/O complexity, we do \emph{not}
update the corresponding pointers of $(k,v)$ pairs in Table~1; this creates the presence of ``spurious''
pointers in Table~1, but Angelino {\it et al}. \cite{agmt-emm-11} explain how to handle the presence of such spurious
pointers while increasing the cost of all other operations by $O(1)$ factors.
\end{enumerate}

All operations above are performed in $O(1)$ time in the worst case with overwhelming probability by the results of Section \ref{sec:cuckoo}. 
Two major issues arise in the above construction. First, the space-usage remains $O(n)$ only if we assume the existence of a garbage-collector for leaked memory, as well as a memory allocation mechanism, both of which must run in $O(1)$ time in the worst case. Without the memory allocation mechanism, inserting $v$ into $L(k)$  cannot be done in $O(1)$ time, and without the garbage collector for leaked memory, space cannot be reused after $\text{remove}$ and $\text{removeAll}$ operations. Second, in order to extract the actual values from a $\text{getAll}(k)$ operation, one must actually traverse the list $L(k)$. Since $L(k)$ may be spread all over memory, this suffers from poor locality. 

We now present our cache-oblivious multimap implementation. Our implementation avoids the need for garbage collection, and circumvents the poor locality 
of the above $\text{getAll}$ operation. We do require a cache-oblivious mechanism
to allocate and deallocate power-of-two sized memory blocks with constant-factor space and I/O overhead; this assumption is theoretically justified by the results of Brodal {\it et al}. \cite{buddysystem}.

\medskip
\noindent \textbf{Amortized Cache-Oblivious Multimaps.}
As in the RAM implementation, we keep two nested cuckoo tables. In Table~1, we store all the $(k,v)$ pairs using each
entire key-value pair as the 
key. With each such pair, we store a count, which identifies an ordinal 
number for this value $v$ associated with this key, $k$, starting
from $0$. For example, if
the keys were (4, Alice), (4, Bob), and (4, Eve), then (4, Alice) might be pair~0, 
(4, Bob) pair~1, and (4, Eve) pair~2, all for the key, 4.

In Table~2, we store all the unique keys.
For each key, $k$, we store a pointer to an array, $A_k$, that stores all the
key-value pairs having key $k$, stored in order
by their ordinal values from Table~1.
With the record for 
a key $k$, we also store $n_k$, the number of pairs
having the key $k$, i.e., the number of key-value pairs in $A_k$.
We assume that each $A_k$ is maintained as an array that supports
amortized $O(1)$-time element access and addition, while maintaining
its size to be $O(n_k)$.

\medskip
\noindent \textbf{Operations.}
\begin{enumerate}
\item $\text{containsKey}(k)$: We perform a lookup for $k$ in Table~2.
\item $\text{containsItem}(k,v)$: We perform a lookup for $(k, v)$ in Table~1.
\item $\text{add}(k,v)$: After ensuring that $(k,v)$ is not already in the multimap by looking it up in Table~1,
we look up $k$ in Table~2, and add $(k,v)$ at index $n_k$
of the array $A_k$, if $k$ is present in this table. 
If there is no key $k$ in Table~2, then we allocate an array, $A_k$, 
of initial constant size. Then we add $(k,v)$ to $A_k[0]$ and add
key $k$ to Table~2.
In either case,
we then add $(k,v)$ to Table~1, giving it ordinal $n_k$,
and increment the value of $n_k$ associated with $k$ in Table~2. This operation may additionally 
require the growth of $A_k$ by a factor of two, which would then necessitate copying all elements 
to the new array location and updating the pointer for $k$ in Table~2.
\item $\text{remove}(k,v)$: We look up $(k,v)$ in Table~1 and get its ordinal 
count, $i$. Then we remove $(k,v)$ from Table~1, and we look up $k$ in Table~2, to 
learn the value of $n_k$ and get a pointer to $A_k$.
If $n_k > 1$, we swap $(k',v')=A_k[n_k-1]$ and $(k,v)=A_k[i]$, and then 
remove the last element of $A_k$.
We update the ordinal value of $(k',v')$ in Table~1 to now
be $i$.
We then decrement the value of $n_k$ associated with $k$ in Table~2.
If this results in $n_k=0$, we remove $k$ from Table~2. 
This operation may additionally require the
shrinkage of the array $A_k$ by a factor of $2$, so as to maintain the $O(n)$ space bound.
\item $\text{getAll}(k)$: We look up $k$ in Table~2, and then list
the contents of the $n_k$ elements stored at the array $A_k$ indexed
from this record.
\item $\text{removeAll}(k)$: For all entries $(k, v)$ of $A_k$, we remove $(k,v)$ from Table~1. We also remove $k$ from Table~2 and deallocate
the space used for $A_k$. As in the RAM implementation, in order to achieve unamortized $O(1)$ I/O cost, we do not
update the pointers of $(k,v)$ pairs in Table~1; this creates the presence of ``spurious''
pointers in Table~1 which are handled the same as in the RAM case.

\end{enumerate}

In terms of I/O performance, $\text{containsKey}(k)$ and  $\text{containsItem}(k,v)$ clearly require $O(1)$ I/Os in the worst case.
$\text{getAll}(k)$ operations use $O(1+n_k/B)$ I/Os in the worst case, because scanning an array of size $n_k$ uses
$O(\lceil n_k/B\rceil)$ I/Os, even though 
we don't know the value of $B$. $\text{removeAll}(k)$ utilizes $O(n_k)$ I/Os in the worst-case with overwhelming probability,
but these can be charged to the insertions of the $n_k$ values associated with $k$, for $O(1)$ amortized I/O cost.
$\text{add}(k,v)$ and $\text{remove}(k,v)$ operations also require $O(1)$ amortized I/Os with overwhelming probability;
the bound is amortized because there is a chance this operation will require a growth or
shrinkage of the array $A_k$, which may require moving all $(k,v)$ values associated with $k$ and updating the corresponding pointers in Table~1.

In the next section, we explain how to deamortize $\text{add}(k, v)$ and $\text{remove}(k, v)$ operations.

\medskip
\noindent \textbf{De-Amortizing the Key-Value Arrays.}
To de-amortize
the array operations, we use a rebuilding 
technique, which is standard in de-amortization methods (e.g.,
see~\cite{kp-daa-98}).

We consider the operations needed for insertions to an array;
the methods for deletions are similar.
The main idea is that we allocate arrays whose sizes are powers
of 2. Whenever an array, $A$, becomes half full, we allocate an array,
$A'$, of double the size and start copying elements $A$ in $A'$.
In particular, we maintain a crossover index, $i_A$, which indicates
the place in $A$ up to which we have copied its contents into $A'$.
Each time we wish to access $A$ during this build phase, we
copy two elements of $A$ into $A'$, picking up at position $i_A$, and updating
the two corresponding pointers in Table~1.
Then we perform the access of $A$, as would would otherwise, except
that if we wish access an index $i<i_A$, then we actually perform 
this access in $A'$. 
Since we copy two elements of $A$ for
every access, we are certain to complete the building of $A'$ 
prior to our needing to allocate a new, even larger array, even if all
these accesses are insertions.
Thus, each access of our array will now complete in worst-case $O(1)$
time with overwhelming probability.
It immediately follows that $\text{add}(k, v)$ and $\text{remove}(k, v)$ operations
run in $O(1)$ worst-case time. 
All time bounds in Table~\ref{tbl:bounds} follow.

\vspace{-4mm}



\section{Conclusion}
In this paper, we have studied fully de-amortized 
dictionary and multimap algorithms that 
support worst-case constant-time operations with 
high or overwhelming probability. At the
core of our result is a ``nested'' cuckoo hash construction, 
in which an inner cuckoo table
is used to support fast lookups into a queue/stash structure 
for an outer cuckoo table, 
as well as a simplified and improved implementation of an 
\emph{atomic stash}, which is related
to the atomic heap or q-heap data structure of 
Fredman and Willard \cite{fredmanwillard}. We gave
fully de-amortized constructions with guarantees that 
hold with high probability in the Practical RAM
model, and with overwhelming probability in  the external-memory
(I/O) model, the standard RAM model, or the AC$^0$ RAM model.

Several interesting questions remain for future work. 
First, lookups in our structure may require four or more 
I/Os in external-memory; 
it would be interesting to develop fully de-amortized structures 
supporting lookups in as few
as two I/Os. 
A prime possibility suited for external memory 
is random-walk cuckoo hashing with two hash functions and super-constant bucket sizes. 
Second, it would be interesting to develop a 
fully-deamortized dictionary for the Practical RAM model where all
operations take $O(1)$ time with \emph{overwhelming} probability.

\subsection*{Acknowledgments}
This research was supported in part by the U.S.~National Science
Foundation, under grants 0713046, 0830403, and 0847968, 
and by an Office
of Naval Research: Multidisciplinary University Research Initiative
(MURI) Award, number N00014-08-1-1015. Justin Thaler is supported by the Department of
Defense (DoD) through the National Defense Science \& Engineering Graduate Fellowship (NDSEG) Program.
Michael Mitzenmacher was supported in part by 
the U.S.~National Science Foundation, under grants 0964473, 0915922, and 0721491.

\bibliographystyle{abbrv}
\bibliography{cuckoo,cuckoo2,extra2,goodrich,range,par}

\begin{thebibliography}{10}

\bibitem{av-iocsrp-88}
A.~Aggarwal and J.~S. Vitter.
\newblock The input/output complexity of sorting and related problems.
\newblock {\em Comm. ACM}, 31:1116--1127, 1987.

\bibitem{Andersson1999337}
A.~Andersson, P.~B. Miltersen, and M.~Thorup.
\newblock Fusion trees can be implemented with {AC0} instructions only.
\newblock {\em Theoretical Computer Science}, 215(1-2):337--344, 1999.

\bibitem{agmt-emm-11}
E.~Angelino, M.~T. Goodrich, M.~Mitzenmacher, and J.~Thaler.
\newblock External-memory multimaps.
\newblock {\em ArXiV ePrints}, abs/1104.5533, 2011.

\bibitem{yuriy_thesis}
Y.~Arbitman.
\newblock Efficient dictionary data structures based on cuckoo hashing.
\newblock MSc Thesis, Department of Computer Science and Applied Mathematics,
  Weizmann Institute of Science, Israel, 2010.

\bibitem{arbitman2009amortized}
Y.~Arbitman, M.~Naor, and G.~Segev.
\newblock {De-amortized cuckoo hashing: Provable worst-case performance and
  experimental results}.
\newblock {\em Automata, Languages and Programming}, pages 107--118, 2009.

\bibitem{ans-bchcw-10}
Y.~Arbitman, M.~Naor, and G.~Segev.
\newblock Backyard cuckoo hashing: Constant worst-case operations with a
  succinct representation.
\newblock In {\em IEEE Symp. on Foundations of Computer Science (FOCS)}, pages
  787--796, 2010.

\bibitem{bdf-cob-00}
M.~A. Bender, E.~D. Demaine, and M.~Farach-Colton.
\newblock Cache-oblivious b-trees.
\newblock In {\em 41st IEEE Symp. on Foundations of Computer Science (FOCS)},
  pages 399--409, 2000.

\bibitem{braverman09}
M.~Braverman.
\newblock Poly-logarithmic independence fools ac0 circuits.
\newblock {\em Electronic Colloquium on Computational Complexity}, 16:3--8,
  2009.

\bibitem{bfj-cost-02}
G.~S. Brodal, R.~Fagerberg, and R.~Jacob.
\newblock Cache oblivious search trees via binary trees of small height.
\newblock In {\em 13th ACM-SIAM Symp. on Discrete Algorithms (SODA)}, pages
  39--48, 2002.

\bibitem{bfv-ecos-08}
G.~S. Brodal, R.~Fagerberg, and K.~Vinther.
\newblock Engineering a cache-oblivious sorting algorithm.
\newblock {\em J. Exp. Algorithmics}, 12:2.2:1--2.2:23, 2008.

\bibitem{buddysystem}
S.~Brodal, D.~Demaine, and I.~Munro.
\newblock Fast allocation and deallocation with an improved buddy system.
\newblock {\em Acta Inf.}, 41:273--291, March 2005.

\bibitem{bc-itvqt-05}
S.~B\"{u}ttcher and C.~L.~A. Clarke.
\newblock Indexing time vs. query time: trade-offs in dynamic information
  retrieval systems.
\newblock In {\em Proc. of 14th ACM Conf. on Information and Knowledge
  Management (CIKM)}, pages 317--318. ACM, 2005.

\bibitem{bcl-himgt-06}
S.~B\"{u}ttcher, C.~L.~A. Clarke, and B.~Lushman.
\newblock Hybrid index maintenance for growing text collections.
\newblock In {\em Proc. of 29th ACM SIGIR Conf. on Research and Development in
  Information Retrieval (SIGIR)}, pages 356--363. ACM, 2006.

\bibitem{clrs-ia-01}
T.~H. Cormen, C.~E. Leiserson, R.~L. Rivest, and C.~Stein.
\newblock {\em Introduction to Algorithms}.
\newblock MIT Press, Cambridge, MA, 2nd edition, 2001.

\bibitem{cp-odiim-90}
D.~Cutting and J.~Pedersen.
\newblock Optimization for dynamic inverted index maintenance.
\newblock In {\em 13th ACM SIGIR Conf. on Research and Development in
  Information Retrieval}, SIGIR '90, pages 405--411. ACM, 1990.

\bibitem{devroye2003cuckoo}
L.~Devroye and P.~Morin.
\newblock {Cuckoo hashing: {Further} analysis}.
\newblock {\em Information Processing Letters}, 86(4):215--219, 2003.

\bibitem{DM03}
L.~Devroye and P.~Morin.
\newblock Cuckoo hashing: Further analysis.
\newblock {\em Information Processing Letters}, 86:215--219, 2003.

\bibitem{dr09}
M.~Dietzfelbinger and M.~Rink.
\newblock Applications of a splitting trick.
\newblock In {\em Proceedings of the 36th International Colloquium on Automata,
  Languages and Programming: Part I}, ICALP '09, pages 354--365, Berlin,
  Heidelberg, 2009. Springer-Verlag.

\bibitem{dw03}
M.~Dietzfelbinger and P.~Woelfel.
\newblock Almost random graphs with simple hash functions.
\newblock In {\em Proceedings of the Thirty-Fifth Annual ACM Symposium on
  Theory of Computing}, STOC '03, pages 629--638, New York, NY, USA, 2003. ACM.

\bibitem{fredmanwillard}
M.~L. Fredman and D.~E. Willard.
\newblock Surpassing the information theoretic bound with fusion trees.
\newblock {\em J. Comput. System Sci.}, 47:424--436, 1993.

\bibitem{flpr-coa-99}
M.~Frigo, C.~E. Leiserson, H.~Prokop, and S.~Ramachandran.
\newblock Cache-oblivious algorithms.
\newblock In {\em 40th IEEE Symp. on Foundations of Computer Science (FOCS)},
  pages 285--298, 1999.

\bibitem{GoodrichMitzenmacherICALParXivVersion}
M.~T. Goodrich and M.~Mitzenmacher.
\newblock Privacy-preserving access of outsourced data via oblivious ram
  simulation.
\newblock {\em CoRR}, abs/1007.1259, 2010.

\bibitem{gcxw-eoli-07}
R.~Guo, X.~Cheng, H.~Xu, and B.~Wang.
\newblock Efficient on-line index maintenance for dynamic text collections by
  using dynamic balancing tree.
\newblock In {\em Proc. of 16th ACM Conf. on Information and Knowledge
  Management (CIKM)}, CIKM '07, pages 751--760. ACM, 2007.

\bibitem{km-uqdac-07}
A.~Kirsch and M.~Mitzenmacher.
\newblock Using a queue to de-amortize cuckoo hashing in hardware.
\newblock In {\em 45th Allerton Conference on Communication, Control, and
  Computing}, pages 751--758, 2007.

\bibitem{kmw-chs-09}
A.~Kirsch, M.~Mitzenmacher, and U.~Wieder.
\newblock More robust hashing: cuckoo hashing with a stash.
\newblock {\em SIAM J. Comput.}, 39:1543--1561, 2009.

\bibitem{k-ss-73}
D.~E. Knuth.
\newblock {\em Sorting and Searching}, volume~3 of {\em The Art of Computer
  Programming}.
\newblock Addison-Wesley, Reading, MA, 1973.

\bibitem{kutzelnigg}
R.~Kutzelnigg.
\newblock Bipartite random graphs and cuckoo hashing.
\newblock {\em DMTCS Proceedings}, 0(1), 2006.

\bibitem{k-ivcha-10}
R.~Kutzelnigg.
\newblock An improved version of cuckoo hashing: Average case analysis of
  construction cost and search operations.
\newblock {\em Mathematics in Computer Science}, 3:47--60, 2010.

\bibitem{lmz-eoict-08}
N.~Lester, A.~Moffat, and J.~Zobel.
\newblock Efficient online index construction for text databases.
\newblock {\em ACM Trans. Database Syst.}, 33:19:1--19:33, September 2008.

\bibitem{lzw-eoimc-06}
N.~Lester, J.~Zobel, and H.~Williams.
\newblock Efficient online index maintenance for contiguous inverted lists.
\newblock {\em Inf. Processing \& Management}, 42(4):916--933, 2006.

\bibitem{ll-eimef-07}
R.~W. Luk and W.~Lam.
\newblock Efficient in-memory extensible inverted file.
\newblock {\em Information Systems}, 32(5):733--754, 2007.

\bibitem{McDiarmid-98-concentration}
C.~McDiarmid.
\newblock Concentration.
\newblock In {\em Probabilistic methods for algorithmic discrete
  mathematics,(M. Habib, C. McDiarmid, J. Ramirez-Alfonsin, B. Reed, Eds.)},
  pages 195--248, Berlin, 1998. Springer.

\bibitem{m-lbstr-96}
P.~B. Miltersen.
\newblock Lower bounds for static dictionaries on {RAMs} with bit operations
  but no multiplication.
\newblock In F.~Meyer and B.~Monien, editors, {\em Int.\ Conf.\ on Automata,
  Languages and Programming (ICALP)}, volume 1099 of {\em LNCS}, pages
  442--453. Springer, 1996.

\bibitem{naor-history}
M.~Naor, G.~Segev, and U.~Wieder.
\newblock {History-independent cuckoo hashing}.
\newblock In {\em Proceedings of ICALP}, pages 631--642. Springer, 2008.

\bibitem{pp08}
A.~Pagh and R.~Pagh.
\newblock Uniform hashing in constant time and optimal space.
\newblock {\em SIAM J. Comput.}, 38:85--96, March 2008.

\bibitem{pr-ch-04}
R.~Pagh and F.~Rodler.
\newblock Cuckoo hashing.
\newblock {\em Journal of Algorithms}, 52:122--144, 2004.

\bibitem{pwyz-coh-10}
R.~Pagh, Z.~Wei, K.~Yi, and Q.~Zhang.
\newblock Cache-oblivious hashing.
\newblock In {\em 29th ACM Symp. on Principles of Database Systems (PODS)},
  pages 297--304, 2010.

\bibitem{kp-daa-98}
S.~Rao~Kosaraju and M.~Pop.
\newblock De-amortization of algorithms.
\newblock In W.-L. Hsu and M.-Y. Kao, editors, {\em Computing and
  Combinatorics}, volume 1449 of {\em LNCS}, pages 4--14. Springer, 1998.

\bibitem{siegel04}
A.~Siegel.
\newblock On universal classes of extremely random constant-time hash
  functions.
\newblock {\em SIAM J. Comput.}, 33:505--543, March 2004.

\bibitem{thorup07}
M.~Thorup.
\newblock On {AC0} implementations of fusion trees and atomic heaps.
\newblock In {\em 14th ACM-SIAM Symposium on Discrete Algorithms (SODA)}, pages
  699--707, 2003.

\bibitem{verbin2010}
E.~Verbin and Q.~Zhang.
\newblock The limits of buffering: a tight lower bound for dynamic membership
  in the external memory model.
\newblock In {\em Proceedings of the 42nd ACM Symposium on Theory of
  Computing}, STOC '10, pages 447--456, New York, NY, USA, 2010. ACM.

\bibitem{v-emads-01}
J.~S. Vitter.
\newblock External memory algorithms and data structures: dealing with massive
  data.
\newblock {\em ACM Comput. Surv.}, 33(2):209--271, 2001.

\bibitem{Willard}
D.~E. Willard.
\newblock Examining computational geometry, van emde boas trees, and hashing
  from the perspective of the fusion tree.
\newblock {\em SIAM J. Comput.}, 29:1030--1049, December 1999.

\bibitem{zm-iftse-06}
J.~Zobel and A.~Moffat.
\newblock Inverted files for text search engines.
\newblock {\em ACM Comput. Surv.}, 38, July 2006.

\end{thebibliography}

\begin{appendix}
\section{Our Atomic Stash Structure}
\label{app:stash}
\medskip\noindent
In this section, we describe our deterministic 
\emph{atomic stash} implementation
of the dictionary ADT.
This structure dynamically
maintains a sparse set, $S$, of
at most $O(w^{1/2})$ key-value pairs, using $O(w)$ words of memory,
so as to support insertions, deletions,
and lookups in $O(1)$ worst-case time, where $w$ is our computer's
word size.
Our construction is valid in the I/O model and the
AC0 RAM model~\cite{thorup07},
that is,
the RAM model minus constant-time multiplication and division,
but augmented with constant-time AC0 instructions that
are included in modern CPU instruction sets.
(See previous work on the atomic heap~\cite{fredmanwillard}
data structure for a 
solution in the standard RAM model, albeit in a way that limits the
size of $S$ to be $O(w^{1/6})$ and is less
efficient in terms of constant factors.)
We assume that keys and values can
each fit in a word of memory.

Our construction builds on the fusion tree and atomic heap
implementations of previous 
researchers~\cite{Andersson1999337,fredmanwillard,thorup07}, but improves the
simplicity and capacity of these previous implementations by
taking advantage of modern CPU instruction sets and the fact that we
are interested here in maintaining only a simple dictionary, rather
than an ordered set.
For instance, our solution avoids any lookups in pre-computed
tables.

\subsection{Components of an Atomic Stash}
\medskip\noindent
Our solution is based on associating
with each key $k$ in $S$, a compressed key, $B_k$, 
of size $w'=\lfloor w^{1/2}\rfloor-1$, which we store as
a representative of $k$.
In addition, for each compressed key, $B_k$,
we store a binary
mask, $M_k$, of equal length, such that
\[
B_k \wedge M_k \not= B_j \wedge M_k,
\]
for $j\not=k$, with $j$ in $S$, where ``$\wedge$'' denotes bit-wise
AND.
That is, all of the masked keys in $S$ are unique.

A critical element in our data structure
is an \emph{associative cache}, $X$, which is stored in a single
word, which we view as being divided into 
$w'$ fields of size
$w'$ each, plus a (high-order)
indicator bit for each field, so $w\ge w'(w'+1)$.
We denote the field with index $i$ in $X$ as $X(i)$ and 
the indicator bit with index $i$ (for field $i$)
in $X$ as $X[i]$, and we assume that
indices are increasing right-to-left and begin with index $1$.
Thus, the bit position of $X[j]$ in $X$ is $j(w'+1)$.
Note that with standard SHIFT, OR, and AND operations,  
we can read or write any field or indicator bit in $X$ in $O(1)$ time given its
index and either a single bit or 
a bit mask, $\vec 1$, of $w'$ 1's.
Each field $X(i)$ is either empty or it stores a compressed key, $B_k$, for 
some key $k$ in $S$.
In addition, we also maintain a word, $Y$, with indices corresponding
to those in $X$, such that $Y(i)$
stores the mask, $M_k$, if $X(i)$ stores the binary key $B_k$.
We also maintain key and value arrays, $K$ and $V$,
such that, if $B_k$ is stored in $X(i)$, then we store the 
key-value pair $(k,v)$ in $K$ and $V$,
so that $K[i]=k$ and $V[i]=v$.
To keep track of the size of $S$,
we maintain a count, $n_S$, which is the number of items in $S$.
Finally, we maintain a ``used'' mask, $U$, which has all 1's in
each field of $X$ that is used. 
That is, $U(i)={\vec 1}$ iff $X(i)$ holds some key, $B_k$, which is
equivalent to saying $Y(i)\not=0$.

The reason we use both a compressed key and a mask for each key
stored in $S$ is that, as we add keys to $S$, we may need to
sometimes expand the function that compresses keys so that the
masked values of keys remain unique, while still fitting in a field of $X$.
Even in this environment, we would like for
previously-compressed keys to still be valid.
In particular, our method generates a sequence of compression
functions, $p_1$, $p_2$, etc., so that $p_i$ returns a $w'$-bit
string whose first $i$ bits are significant (and can be either $0$ or
$1$) and whose remaining bits are all $0$'s.
In addition, if we let $M^d$ denote a bit mask
with $d$ significant $1$ bits and $w'-d$ following $0$ bits, then,
for $d\ge 2$,
\[
p_d(k) \wedge M^{d-1} = p_{d-1}(k) ,
\]
for any key $k$.

\subsection{Operations in an Atomic Stash}
\medskip\noindent
Let us now describe how we perform the various update and query operations on an
atomic stash.
Assume we may use the following primitive operations in our
methods:
\begin{itemize}
\item
A binary operator,
``$\oplus$,'' 
which denotes the bit-size XOR operation.
\item
DUPLICATE$(B)$: return a
word $Z$ having the binary key $B$ (of size $w'$)
stored in each of its fields.
(Note: we can implement DUPLICATE either using a single multiplication
or using $O(1)$ instructions of a modern CPU.)
\item
VecEQ$(W,B)$: given a word $W$ and a binary key $B$
(of size $w'$),
set the indicator bit $W[i]=1$ iff $B=W(i)$.
This operation can be implemented using standard AC0
operations~\cite{Andersson1999337,thorup07}.
\item
MSB$(W)$: return the index of the most significant $1$ bit in the word
$W$, or $0$ if $W=0$. 
As Thorup observed~\cite{thorup07},
this operation can be implemented, for example, by converting $W$ to
floating point and returning the exponent plus 1.
\end{itemize}

We perform a getIndex$(k)$ operation, which returns the index of key $k$ 
in $X$, or $0$ if $k$ is not a key in $S$, as follows.
We assume in this method that we have access to the current
compression function, $p_d$.
\begin{itemize}
\item getIndex$(k)$:
\begin{algorithmic}
\STATE $B\leftarrow p_d(k)$
\STATE $Z\leftarrow$ DUPLICATE$(B)$
\STATE $T\leftarrow Y\wedge (X\oplus Z)$
\STATE $R\leftarrow U\oplus T$
\STATE VecEQ$(R,{\vec 1})$
\STATE \textbf{return} MSB$(R)/(w'+1)$
\end{algorithmic}
\end{itemize}

The correctness of this method
follows from the fact that
$B$ is the key in $S$ at index $i$ associated with $k$ 
iff $T(i)$ is all 0's and index $i$
is being used to store a key from $S$, since all masked keys in $S$
are unique
(and if $B$ is not in $X$, then MSB$(R)=0$).
Also, if one desires that we should avoid integer division,
then we can define $w'$ so that $(w'+1)$ is a power of $2$, 
while keeping $w\ge w'(w'+1)$,
so that above division can be done as a SHIFT.
In addition, note that we can implement a get$(k)$ operation, by
performing a call to getIndex$(k)$ and, if the index, $i$, returned
is greater than $0$, then returning $(K[i],V[i])$.

To remove the key-value pair from $S$ associated with a key $k$ in
$S$, we perform the following operation:
\begin{itemize}
\item remove$(k)$:
\begin{algorithmic}
\STATE $i\leftarrow$ get$(k)$
\STATE $X(i)\leftarrow X(n_S)$
\STATE $Y(i)\leftarrow Y(n_S)$
\STATE $U(n_S)\leftarrow 0$
\STATE $K[i]\leftarrow K[n_S]$ 
\STATE $V[i]\leftarrow V[n_S]$ 
\STATE $n_S \leftarrow n_S - 1$
\end{algorithmic}
\end{itemize}

Our insertion method, which follows, assumes that we have 
access to the current compression function, $p_d$,
as well as the mask, $M^d$, of $d$ 1's followed by $(w'-d)$ 0's.
We also assume we know the current value of (the global variable)
$d$ and that we have a
function, P\_EXPAND$(k_1,k_2)$, which takes two distinct keys, $k_1$ and $k_2$,
and expands the compression function from $p_d$ to $p_{d+1}$ so that
$p_{d+1}(k_1)\not= p_{d+1}(k_2)$.
This method also defines $M^{d+1}$ to consist of $(d+1)$ significant
$1$ bits and $w'-(d+1)$ trailing $0$'s, and it increments $d$.
\begin{itemize}
\item add$(k,v)$:
\begin{algorithmic}
\STATE $i\leftarrow$ get$(k)$
\IF[We had a collision.] {$i>0$}
\STATE P\_EXPAND$(k,K[i])$ \hspace{1em} \{Also increments $d$\}
\STATE $X(i) \leftarrow p_{d}(K[i])$
\STATE $Y(i) \leftarrow M^{d}$
\ENDIF
\STATE $n_S \leftarrow n_S + 1$
\STATE $X(n_S)\leftarrow p_d(k)$
\STATE $Y(n_S)\leftarrow M^d$
\STATE $U(n_S)\leftarrow {\vec 1}$
\STATE $K[n_S]\leftarrow k$
\STATE $V[n_S]\leftarrow v$
\end{algorithmic}
\end{itemize}

Since the compressed keys in $S$
form a set of unique masked keys, the new key,
$k$, will collide with at most one of them, even when compressed. 
So, it is sufficient
in this case that we expand the compression function
so that it distinguishes $k$ and this one colliding field.
Thus, a simple inductive argument, we maintain the property that
all the masked keys in $S$ are unique.

\subsection{Compressing Keys and De-Amortization}
\medskip\noindent
Let us now describe how we compress keys and expand the compression
function, as well as how we perform the necessary de-amortization so
that the size of compressed keys is never more than $w'$.

Our method makes use of the following primitive AC0
operation~\cite{Andersson1999337,thorup07},
which can also be computed from
included primitives in modern CPU instruction sets.
\begin{itemize}
\item
SELECT$(W,k)$: Given a word $W$ and key $k$, 
the fields of $W$ are viewed as bit pointers. A length $w'$ bit
string, $B$, is returned so that the $i$-th bit 
of $B$ equals the $W(i)$-th bit of $k$.
\end{itemize}

Our compression function is encoded in terms of the counter, $d$, and
a word, $W$, that encodes the bits to be selected from keys, so that
$W(i)$ is the index of the $i$th bit of $k$ in the output of
$p_d(k)$, for $i\le d$. For $i>d$, $W(i)=0$.
Thus, we
define $p_d$ simply as follows.
\begin{itemize}
\item
$p_d(k)$:
\begin{algorithmic}
\STATE \textbf{return} SELECT$(W,k)$
\end{algorithmic}
\end{itemize}

Thus, we also have a simple definition for P\_EXPAND:
\begin{itemize}
\item
P\_EXPAND$(k_1,k_2)$:
\begin{algorithmic}
\STATE $d\leftarrow d+1$
\STATE $W(d)\leftarrow {\rm MSB}(k_1\oplus k_2)$
\STATE $M^d \leftarrow M^{d-1} \vee (1\ \  {\rm SHIFT}\ \  d)$
\end{algorithmic}
\end{itemize}

Note that in the way we are
using the expansion function, $p_d$, we only ever add bits to
the ends of our compressed keys.
We never remove bits. This allows keys compressed under previous
instantiations of the compression function to continue to be used.
But this also implies that, as we continue to add keys to $S$,
we might run out of fields to use, even if we keep the size, $n_S$, 
of $S$ to be at most, say, $w'/2$.
Of course, we could use a standard amortized rebuilding
scheme to maintain $d$ to be at most $w'$, but this would require
using amortization instead of achieving true worst-case constant-time
bounds for our updates and lookups.

As a de-amortization technique, therefore,
let us revise our construction, so that,
whenever $d > w'/2$, we create a new, initially empty,
atomic stash.
For each additional atomic-stash 
operation after this initialization, we remove
two items from the old atomic stash and add them to this
new atomic stash.
In performing accesses and updates during this time, we also keep
a crossover index, $x$, so that references to fields and indicator indices less
than or equal to $x$ are done in the new stash and references to fields
and indicator indices greater than $x$ are done in the old stash.
Thus, after $w'/2$ additional operations, we can
discard the old stash (for which $d\le w'$) and fully replace it with the
new one (for which $d\le w'/2$).
Therefore, so long as $n_S\le w'/2$, we can perform all insertion,
deletion, and lookup operations in worst-case $O(1)$ time, which in
the I/O model corresponds to $O(1)$ I/Os.

\end{appendix}

\end{document}